\documentclass[11pt]{amsart}
\pdfoutput=1

\usepackage{amssymb}
\usepackage{amsthm}
\usepackage{amsfonts}
\usepackage{amsmath}
\usepackage{pmboxdraw}
\usepackage{verbatim} 
\usepackage{graphicx}
\usepackage{color}
\usepackage[colorlinks=true, citecolor=blue, filecolor=black, linkcolor=black, urlcolor=black]{hyperref}
\usepackage[normalem]{ulem}
\usepackage{subcaption}
\usepackage{bbm}
\usepackage{bm}
\usepackage{mathtools}
\mathtoolsset{showonlyrefs}
\usepackage{todonotes}
\usepackage{enumitem}
\usepackage{tikz-cd}

\newcommand{\RN}[1]{%
	\textup{\uppercase\expandafter{\romannumeral#1}}%
}

\def\bfs{\boldsymbol}

\def\pa{\partial}

\def\wh{\widehat}
\def\wt{\widetilde}

\def\Re{ \mathrm{Re}}

\def\C{\mathbb{C}}

\def\P{\mathbf{P}}
\def\R{\mathbb{R}}

\def\bigO{{\mathcal O}}

\newcommand{\Pf}{{\textup{Pf}}}
\newcommand{\erfc}{\operatorname{erfc}}

\newcommand{\bfR}{\mathbf{R}}

\newcommand{\bfkappa}{{\bm \varkappa}}

\newcommand{\re}{\operatorname{Re}}
\newcommand{\im}{\operatorname{Im}}

\newcommand{\Prob}{{\mathbb{P}}}

\theoremstyle{plain}
\newtheorem*{thm*}{Theorem}
\newtheorem{thm}{Theorem}[section]
\newtheorem{lem}[thm]{Lemma}
\newtheorem{lemma}[thm]{Lemma}

\newtheorem{prop}[thm]{Proposition}
\newtheorem*{prop*}{Proposition}
\newtheorem*{lem*}{Lemma}
\newtheorem{remark}[thm]{Remark}

\theoremstyle{definition}
\newtheorem*{eg*}{Example}
\newtheorem*{egs*}{Examples}
\newtheorem*{def*}{Definition}
\newtheorem*{Q*}{Question}

\theoremstyle{remark}
\newtheorem*{rmk*}{Remark}
\newtheorem*{rmks*}{Remarks}

\newcommand{\abs}[1]{\lvert#1\rvert}
\numberwithin{equation}{section}

\begin{document}
\title[Symplectic induced Ginibre ensemble in the almost-circular regime]{On the almost-circular symplectic induced Ginibre ensemble}
\author{Sung-Soo Byun}
\address{School of Mathematics, Korea Institute for Advanced Study, 85 Hoegiro, Dongdaemun-gu, Seoul 02455, Republic of Korea}
\email{sungsoobyun@kias.re.kr} 

\author{Christophe Charlier}
\address{Centre for Mathematical Sciences, Lund University, 22100 Lund, Sweden}
\email{christophe.charlier@math.lu.se}



\begin{abstract}
We consider the symplectic induced Ginibre process, which is a Pfaffian point process on the plane. Let $N$ be the number of points. We focus on the almost-circular regime where most of the points lie in a thin annulus $\mathcal{S}_{N}$ of width $O(\frac{1}{N})$ as $N \to \infty$. Our main results are the scaling limits of all correlation functions near the real axis, and also away from the real axis. Near the real axis, the limiting correlation functions are Pfaffians with a new correlation kernel, which interpolates the limiting kernels in the bulk of the symplectic Ginibre ensemble and of the anti-symmetric Gaussian Hermitian
ensemble of odd size. Away from the real axis, the limiting correlation functions are determinants, and the kernel is the same as the one appearing in the bulk limit of almost-Hermitian random matrices. Furthermore, we obtain precise large $N$ asymptotics for the probability that no points lie outside $\mathcal{S}_{N}$, as well as of several other ``semi-large" gap probabilities. 
\end{abstract}


\maketitle
\vspace{-0.5cm} \noindent


\section{Introduction and main results}
The symplectic induced Ginibre ensemble is the Pfaffian point process for $N$ points $\bfs{\zeta}= \{ \zeta_j \}_{j=1}^N$ on the plane whose joint probability distribution $\P_N$ is given by 
\begin{equation}\label{Gibbs}
d\P_N(\boldsymbol{\zeta}) = \frac{1}{N!Z_N} \prod_{1 \leq j<k \leq N} \abs{\zeta_j-\zeta_k}^2 \abs{\zeta_j-\overline{\zeta}_k}^2 \prod_{j=1}^{N} \abs{\zeta_j-\overline{\zeta}_j}^2 e^{ -N  Q_{N}(\zeta_j) } \,  dA(\zeta_j),
\end{equation}
where $dA(\zeta):=d^2\zeta/\pi$, and $Z_N$ is the normalisation constant. Here the potential $Q_N$ depends strongly on $N$ and is given by 
\begin{equation} \label{Q aN bN}
 Q_N(\zeta):=a_N |\zeta|^2-2 b_N \log |\zeta|, \qquad a_N=\frac{N}{\rho^2}, \qquad b_N=\frac{N}{\rho^2}-1,
\end{equation}
where $\rho \in (0,\infty)$ is independent of $N$. If $b_N$ is replaced by $0$, then \eqref{Gibbs} is the eigenvalue distribution of a symplectic Ginibre matrix, i.e. a Gaussian random matrix with quaternion entries \cite{ginibre1965statistical}. More generally, for integer values of $b_{N}$, \eqref{Gibbs} is the eigenvalue distribution of a symplectic Ginibre matrix conditioned to have $b_{N}$ zero eigenvalues \cite{MR2180006} (such matrices are called symplectic induced Ginibre matrices \cite{MR2881072}). 

The ensemble \eqref{Gibbs} also admits a statistical mechanics interpretation \cite{kiessling1999note,forrester2016analogies} as a two-dimensional Coulomb gas with an additional complex conjugation symmetry. 
From this point of view, the parameters $a_N$ and $b_N$ play different roles in the statistics of \eqref{Gibbs}: $a_N$ corresponds to the repulsion between the points and $\infty$, whereas $b_N$ corresponds to the repulsion between the points and the origin. Since $a_{N}$ and $b_{N}$ are of order $N$, these repulsions from $0$ and $\infty$ are much stronger than the point-point repulsion. As a consequence, for large $N$ the points are confined in a thin annulus $\mathcal{S}_N$ with high probability. $\mathcal{S}_{N}$ is called the droplet, and by \cite[Theorem 3.1]{MR2934715} and \cite[Section \RN{4}.5]{ST97} it is given by
\begin{figure}
	\centering
	\includegraphics[width=\textwidth]{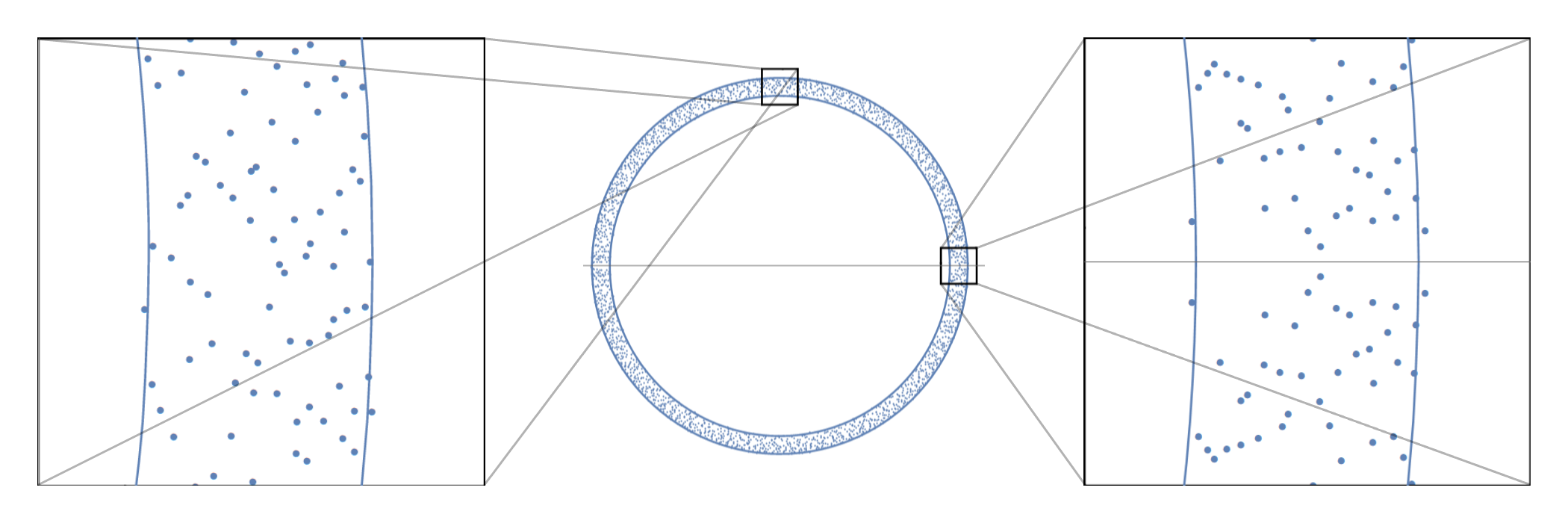}
	\caption{ Illustration of $(\bfs{\zeta},\overline{\bfs{\zeta}})$, where $\bfs{\zeta}$ is drawn from the symplectic induced Ginibre ensemble in the almost-circular regime, with $N=1000$ and $\rho = 10$. The zooms are taken near $1$ and $i$. (We are not aware of a method to simulate \eqref{Gibbs} for $b_{N} \neq 0$, so the above figure has been generated in a somewhat naive way, is inexact, and therefore should be taken with a grain of salt.) 
} 
\label{Fig_AUE}
\end{figure}
\begin{equation} \label{droplet annulus}
\mathcal{S}_N:=\{ \zeta \in \C : r_1 \le |\zeta| \le r_2 \},\qquad r_1=\sqrt{\frac{b_N}{a_N}}, \qquad r_2=\sqrt{ \frac{2+b_N}{a_N} }. 
\end{equation}
With the choice \eqref{Q aN bN} of $a_{N}$ and $b_{N}$, it follows that
\begin{equation} \label{radius}
	r_1 = 1-\frac{\rho^2}{2N}+O\Big(\frac{1}{N^{2}}\Big), \qquad  r_2   = 1+\frac{\rho^2}{2N}+O\Big(\frac{1}{N^{2}}\Big), \qquad \mbox{as } N \to  \infty,
\end{equation}
and therefore $\mathcal{S}_N$ is a thin annulus of width $\frac{\rho^2}{N}+O(N^{-2})$ as $N \to \infty$, see Figure~\ref{Fig_AUE}. 
For this reason, we will refer to the choice of parameters \eqref{Q aN bN}---which is the focus of this paper---as the \emph{almost-circular regime}. 

In this work $N$ is large and $\rho$ is fixed. As an interesting aside, we mention that in the other regime where $\rho \to 0$ while $N$ is kept fixed, the droplet is also thin, and by taking formally $\rho=0$, \eqref{Gibbs} becomes the eigenvalue distribution on the unit circle of a symplectic unitary random matrix, see \cite[Section 2.6]{forrester2010log}.

Planar ensembles with thin droplets were introduced in the works \cite{MR1634312, fyodorov1997almost, MR1431718} on the complex elliptic Ginibre ensemble and were recently studied in \cite{AB21, byun2021random} for more general random normal matrix models. The work \cite{AB21} treats universality questions for general bandlimited point processes using Ward's equation, while \cite{byun2021random} deals with almost-circular ensembles associated with general radially symmetric potentials. In the earlier works \cite{MR1634312, fyodorov1997almost, MR1431718,AB21, byun2021random}, the considered point processes are determinantal. A major difference with our case is that \eqref{Gibbs} is Pfaffian.

We also mention that if $Q_{N}$ is replaced by $2(1+\frac{L+1}{N})\log(1+|\zeta|^{2})-\frac{4L}{N}\log |\zeta|$, then \eqref{Gibbs} is called the symplectic induced spherical ensemble and was studied in \cite{MR3612266}.

In this work, we study scaling limits and gap probabilities of \eqref{Gibbs} as $N \to \infty$ while $\rho$ is kept fixed. Our results can be summarized as follows.

\begin{enumerate}[label=(\roman*)]
	\item \textbf{(Scaling limits)}
	One might expect from Figure \ref{Fig_AUE} that \eqref{Gibbs} will enjoy different limiting correlation structures as $N \to \infty$ depending on whether we look at the statistics near the real line or not. Our results confirm this expectation. Theorem~\ref{thm:main} (a) deals with the limiting correlation structure of \eqref{Gibbs} near a point $p$ on the unit circle, $p \neq -1, 1$, while Theorem~\ref{thm:main} (b) deals with the other cases $p=\pm 1$. For $p=\pm 1$, the Pfaffian structure is preserved in the limit, and it involves a new skew pre-kernel $\kappa^{\mathbb{R}}$ (see \eqref{kappa Wronskian} below) which interpolates between the limiting pre-kernels in the bulk of the symplectic Ginibre ensemble and of the anti-symmetric Gaussian Hermitian
ensemble of odd size (see Remark \ref{remark: mathcing with symplectic Ginibre} and Proposition \ref{Cor_Chiral lim}). For the other case $|p|=1$, $p \neq -1, 1$, we find that the Pfaffian structure of \eqref{Gibbs} simplifies in the limit and becomes determinantal, see \eqref{RNk complex case}--\eqref{def of RkC new}. The limiting correlation kernel $K^\C$ is given by \eqref{K AH} and already appeared in \cite{MR1634312, fyodorov1997almost, MR1431718, akemann2016universality,MR4030288, AB21, byun2021random}, but the way it arises in this work, namely as the large $N$ limit of the kernel of a Pfaffian point process, is new to our knowledge. In Theorem \ref{thm:main} (c), we study a transition between the two families of correlation functions associated with $\kappa^{\mathbb{R}}$ and $K^\C$.
	\smallskip 
	\item\label{item gap} \textbf{(Gap probabilities)} For typical configurations, most of the points of \eqref{Gibbs} lie in $\mathcal{S}_{N}$. In fact, by analogy with e.g. \cite[Eq.(70)]{MR3450566}, we expect about $\sim \sqrt{N}$ of the $\zeta_{j}$'s to lie slightly outside $\mathcal{S}_{N}$. In Theorem~\ref{Thm_partition functions} below, we obtain precise asymptotics for the probability $\mathbb{P}_{N}^{12}$ that \textit{all} $\zeta_{j}$ lie in $\mathcal{S}_{N}$, up and including the term of order $1$. We also obtain similar results for the probability $\mathbb{P}_{N}^{1}$ that $\min_{j \in \{1,\ldots,N\}}|\zeta_{j}| \geq r_{1}$, and for the probability $\mathbb{P}_{N}^{2}$ that $\max_{j \in \{1,\ldots,N\}}|\zeta_{j}| \leq r_{2}$. 
\end{enumerate}

In the following subsections, we give more background and state our main results.

\subsection{Scaling limits}
An interesting feature of \eqref{Gibbs} is that it is integrable and provides one of the few known examples of Pfaffian point processes in the plane \cite{MR1928853}.  The $k$-point correlation function of \eqref{Gibbs} is defined by
\begin{equation}\label{bfRNk def}
	\bfR_{N,k}(\zeta_1,\dots, \zeta_k) := \frac{N!}{(N-k)!} \int_{\C^{N-k}} \P _N(\bfs{\zeta}) \prod_{j=k+1}^N  dA(\zeta_j).
\end{equation}
Since \eqref{Gibbs} is a Pfaffian point process, all correlation functions can be expressed as Pfaffians involving a (skew) pre-kernel $\bfkappa_N$. More precisely, we have
\begin{equation} \label{bfR Pfa}
	\bfR_{N,k}(\zeta_1,\dots, \zeta_k) =  \Pf \Big[ 
	e^{ -\frac{N}{2}(Q_N(\zeta_j)+Q_N(\zeta_l)) } 
	\begin{pmatrix} 
		\bfkappa_N(\zeta_j,\zeta_l) & \bfkappa_N(\zeta_j,\bar{\zeta}_l)
		\smallskip 
		\\
		\bfkappa_N(\bar{\zeta}_j,\zeta_l) & \bfkappa_N(\bar{\zeta}_j,\bar{\zeta}_l) 
	\end{pmatrix}  \Big]_{ j,l=1 }^k \prod_{j=1}^{k} (\bar{\zeta}_j-\zeta_j),
\end{equation}
where
\begin{equation} \label{bfkappaN GN}
	\begin{split}
		\bfkappa_N(\zeta,\eta):=\boldsymbol{G}_N(\zeta,\eta)-\boldsymbol{G}_N(\eta,\zeta),
	\end{split}
\end{equation}
and $\boldsymbol{G}$ is given in terms of the standard $\Gamma$ function by
\begin{equation} 
	\begin{split} \label{GN}
		\boldsymbol{G}_N(\zeta,\eta):=  \sqrt{\pi} \Big( \frac{a_NN}{2} \Big)^{b_NN+\frac32}  \sum_{k=0}^{N-1} \frac{ ( \sqrt{\frac{a_NN}{2}}\, \zeta )^{2k+1} }{  \Gamma(k+\tfrac{3}{2}+\tfrac{b_N N}{2}) }  \sum_{l=0}^k \frac{ ( \sqrt{ \frac{a_NN}{2} } \eta )^{2l} }{ \Gamma(l+\frac{b_N N }{2}+1) }.
	\end{split}
\end{equation}
We prove \eqref{bfR Pfa}--\eqref{GN} in Lemma~\ref{Lem_bfkappa exp} below. In this subsection, we obtain scaling limits of all correlation functions $\{\bfR_{N,k}\}_{k=1}^{\infty}$ as $N \to \infty$ at different points of the droplet. 

The problem of deriving scaling limits for correlation functions of point processes is a classical topic in random matrix theory in connection with the local universality conjecture. This conjecture asserts, roughly speaking, that the limiting correlations between the points should only depend on the symmetry class of the random matrix model and on the points around which these correlations are studied, see e.g. \cite{kuijlaars2011universality} for a survey.

For planar Pfaffian point processes, such problems began to be addressed in the pioneering works of Mehta \cite{Mehta} and Kanzieper \cite{MR1928853}, who studied the limiting local statistics of the symplectic Ginibre ensemble at the origin. In recent years, there has been a growing number of works in that direction. To name a few, the limiting local statistics of the symplectic Ginibre ensemble at the real bulk/edge have been studied in \cite{akemann2021scaling, Lysychkin,akemann2019universal}. 
These results have been extended to the elliptic Ginibre ensemble, see \cite{akemann2021skew} for the limiting kernel at the origin, and \cite{byun2021universal} for the limiting kernel anywhere on the real line. 
For the elliptic Ginibre ensemble in the almost-Hermitian regime, the scaling limits for the kernel at the origin and at the real edge were discovered in \cite{MR1928853} and \cite{MR3192169} respectively. 
These scaling limits have then been extended to the entire real bulk/edge in \cite{byun2021wronskian}.
Similar problems were studied in \cite{khoruzhenko2021truncations} for the truncated symplectic ensemble and in \cite{byun2021wronskian} for the Ginibre ensemble with boundary confinements. We also refer to \cite{akemann2021scaling, MR2180006, MR2302902, MR3066113} for various results on scaling limits of correlation kernels in the context of Mittag-Leffler ensembles, Laguerre ensembles, and products of Ginibre matrices.

For the point process \eqref{Gibbs}, it is natural to expect important differences between the limiting statistics near the real line and those away from the real line. Indeed, the factor $\abs{\zeta_j-\overline{\zeta}_j}^2$ implies that the $\zeta_{j}$'s repel from the real line. 
Furthermore, when the local statistics around the real line are considered, both interaction terms $|\zeta_j-\zeta_k|^2$ and $|\zeta_j-\bar{\zeta}_k|^2$ in \eqref{Gibbs} should contribute as $N \to \infty$. One therefore expect a limiting Pfaffian structure to emerge in the large $N$ limit. On the other hand, when the local statistics are considered away from the real line, only one of the terms $|\zeta_j-\zeta_k|^2$ and $|\zeta_j-\bar{\zeta}_k|^2$ makes a non-trivial contribution; the other one behaves like a constant background charge which leads to a trivial contribution in the large $N$ limit. Therefore, away from the real axis, the limiting local statistics are expected to be determinantal and to be the same as the ones appearing in the random normal matrix model. We now state our first main results, which confirm these expectations.

\begin{thm}\label{thm:main}
Let $\rho \in (0,\infty)$ and $k \in \mathbb{N}_{>0}$ be fixed, let $Q_N$ be as in \eqref{Q aN bN}, and let $\gamma_N:= \sqrt{2} \, \rho/N$.
\begin{itemize}
\item[(a)] \textbf{\textup{Scaling limit of $\bfR_{N,k}$ away from the real line.}} \newline
Let $p:=e^{i\theta}$, where $\theta \in [0,2\pi) \setminus \{0,\pi \}$. As $N \to \infty$, we have 
\begin{align}\label{RNk complex case}
\hspace{1.2cm} \gamma_N^{2k} \,  \bfR_{N,k}\Big(p+p\gamma_{N}z_{1},\dots,p+p\gamma_{N}z_{k}\Big) = R_k^{\mathbb{C}}(z_1,\dots,z_k)+o(1),
\end{align}
uniformly for $z_{1},\ldots,z_{k}$ in compact subsets of $\mathbb{C}$, where
\begin{align}
& R_k^{\mathbb{C}}(z_1,\dots,z_k) := \det \Big[ e^{-|z_j|^2-|z_l|^2}\, K^{\mathbb{C}}(z_j,z_l) \Big]_{j,l=1}^k, \label{def of RkC new} \\
& K^{\mathbb{C}}(z,w):= \frac{e^{2z\bar{w}} }{2}\Big( \erfc(z+\bar{w}-2a)-\erfc(z+\bar{w}+2a) \Big), \quad a:= \frac{\rho}{2\sqrt{2}}. \label{K AH}
\end{align}
\item[(b)] \textbf{\textup{Scaling limit of $\bfR_{N,k}$ near the real line.}} \newline
Let $p = p_N:= e^{i\theta_N}$, where $\theta_N=\frac{\sqrt{2}\rho}{N} t$ and $t \in \R$ is fixed. As $N \to \infty$, we have 
\begin{align}
\gamma_N^{2k} \,  \bfR_{N,k}\Big(p+p\gamma_{N}z_{1},\dots,p+p\gamma_{N}z_{k}\Big) = R_k^{\mathbb{R}}(z_1+it,\dots,z_k+it) + o(1), \label{RNk near real}
\end{align}
uniformly for $z_{1},\ldots,z_{k}$ in compact subsets of $\mathbb{C}$, where 
\begin{align}
& \hspace{1cm} R_k^{\mathbb{R}}(z_1,\dots,z_k):=\Pf \Big[ e^{-|z_j|^2-|z_l|^2} 
		\begin{pmatrix} 
			\kappa^{\mathbb{R}}(z_j,z_l) & \kappa^{\mathbb{R}}(z_j,\bar{z}_l)
			\smallskip \label{def of RkR} \\
\kappa^{\mathbb{R}}(\bar{z}_j,z_l) & \kappa^{\mathbb{R}}(\bar{z}_j,\bar{z}_l) 
		\end{pmatrix} 
		\Big]_{j,l=1}^k\prod_{j=1}^k (\bar{z}_j-z_j) , \\
& \hspace{1cm} \kappa^{\mathbb{R}}(z,w):=\sqrt{\pi} e^{z^2+w^2} \Big( \int_{-a}^a W(f_{w},f_{z})(u) \, du +f_w(a)f_z(-a) -f_z(a)f_w(-a) \Big), 
\label{kappa Wronskian}
\end{align}
$W(f,g):=fg'-gf'$ is the Wronskian, $a:= \frac{\rho}{2\sqrt{2}}$ and 
\begin{equation} \label{fz a}
f_z(u):=\tfrac12 \erfc(\sqrt{2}(z-u)). 
\end{equation}
Statement (b) above also holds with $p = p_N:= -e^{i\theta_N}$.
\smallskip 
\item[(c)] \textbf{\textup{The $R_k^{\mathbb{R}}$-to-$R_k^{\mathbb{C}}$ transition}} 
\newline
As $t \to \infty$, we have 
\begin{equation} \label{Pf to det}
R_k^{\mathbb{R}}(z_1+it,\dots,z_k+it)= R_k^{\mathbb{C}}(z_1,\dots,z_k) +o(1),
\end{equation}
uniformly for $z_1,\dots, z_k$ in compact subsets of $\C$.

\end{itemize}
\end{thm}

\begin{remark}
The $1$-point functions  
\begin{align}
&R_{1}^\C(z)= \frac12 \Big(\erfc(z+\bar{z}-2a)-\erfc(z+\bar{z}+2a)\Big), \label{R1 C}
\\
& R^{\mathbb{R}}_{1} (z)= \sqrt{\pi}  (\bar{z}-z) e^{(z-\bar{z})^2} \Big( \int_{-a}^a W(f_{\bar{z}},f_{z})(u) \, du +f_{\bar{z}}(a)f_z(-a) -f_z(a)f_{\bar{z}}(-a) \Big), \label{R1 R}
\end{align}
are represented in Figure~\ref{Fig_R Pfaffian complex} 
for several choices of $a$.
\end{remark}

\begin{figure}[h!]
	\begin{subfigure}{0.3\textwidth}
		\begin{center}	
			\includegraphics[width=\textwidth]{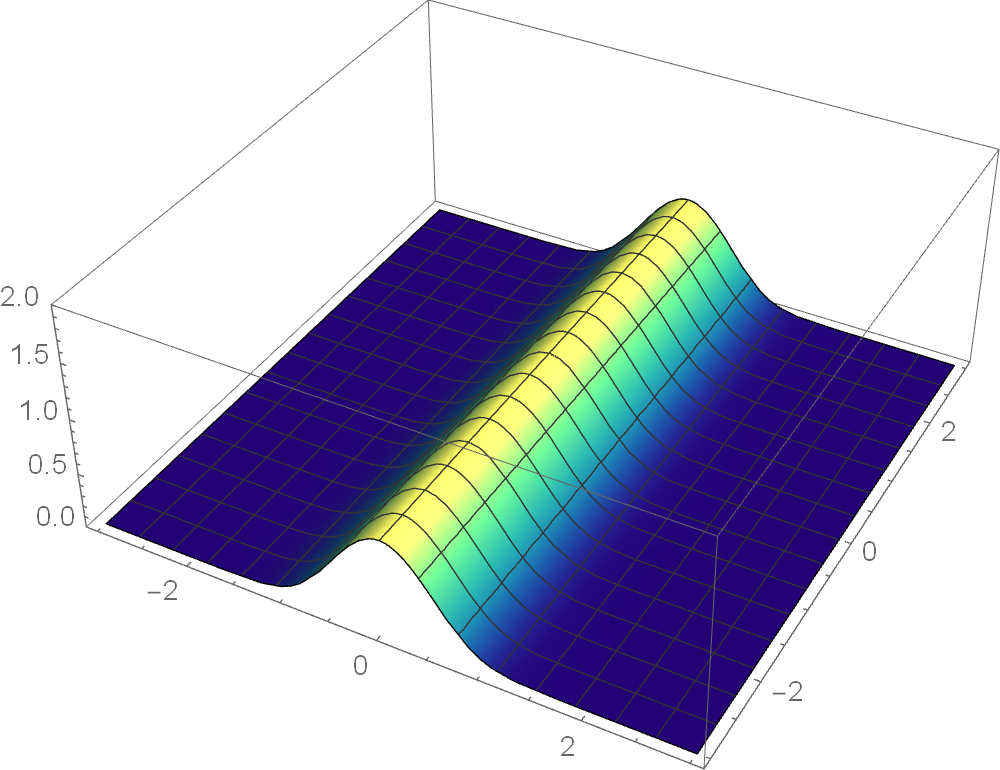}
		\end{center}
		\subcaption{$a=\frac12$}
	\end{subfigure}	
	\begin{subfigure}[h]{0.3\textwidth}
		\begin{center}
			\includegraphics[width=\textwidth]{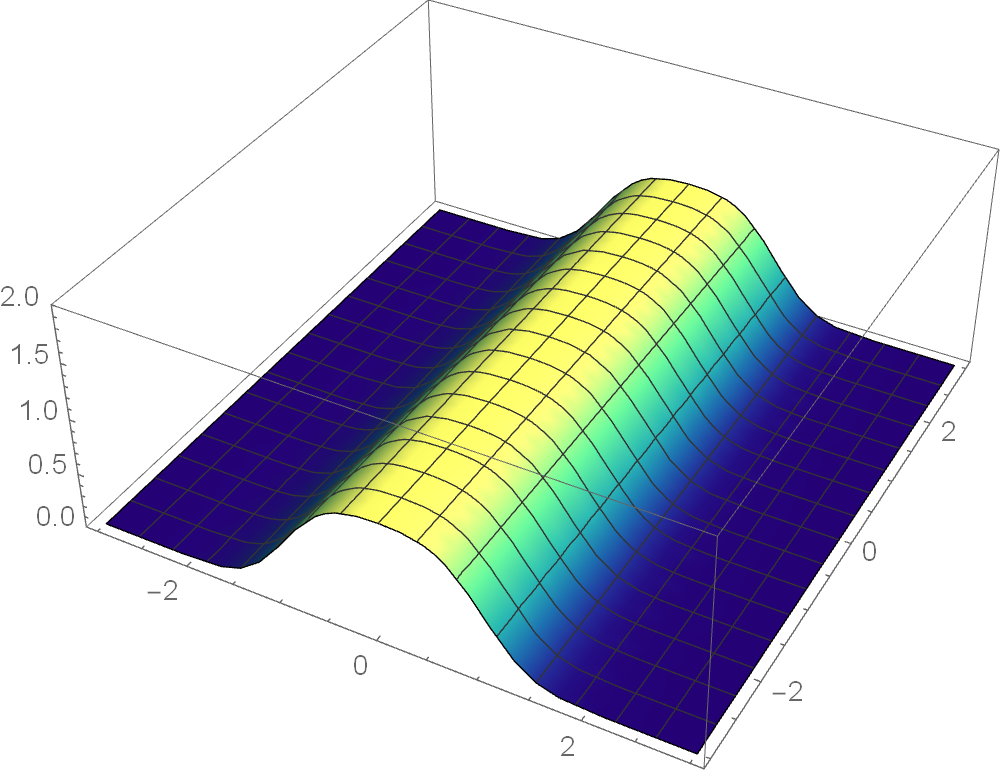}
		\end{center} \subcaption{$a=1$}
	\end{subfigure}
	\begin{subfigure}[h]{0.3\textwidth}
		\begin{center}
			\includegraphics[width=\textwidth]{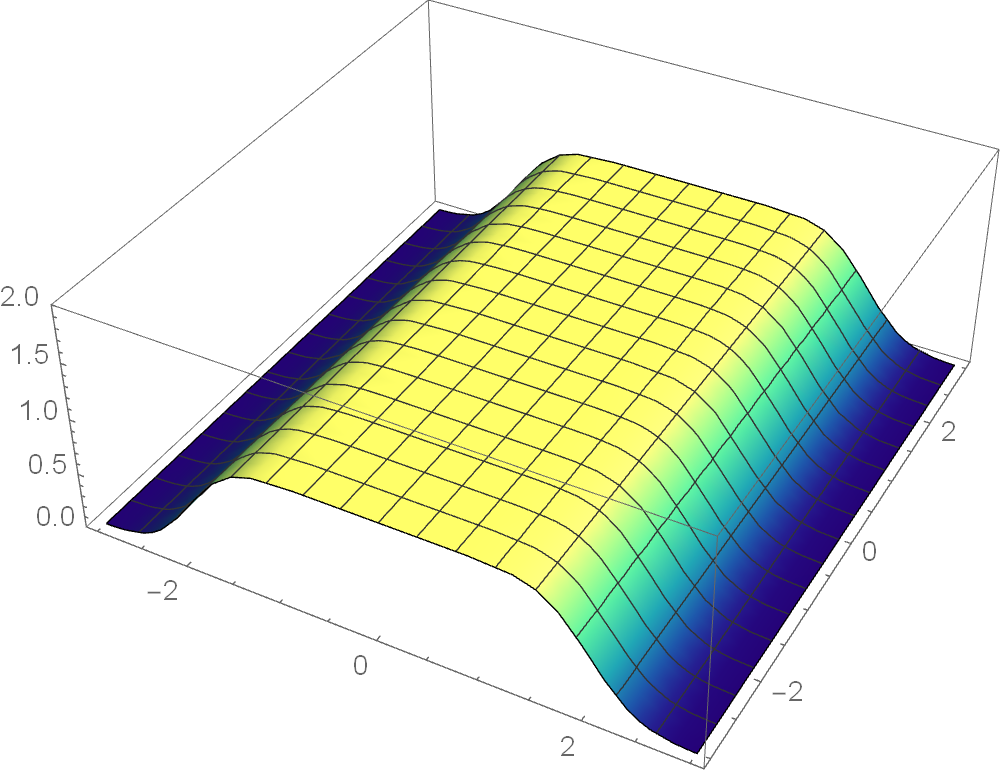}
		\end{center} \subcaption{$a=2$}
	\end{subfigure}
	
	\begin{subfigure}{0.3\textwidth}
		\begin{center}	
			\includegraphics[width=\textwidth]{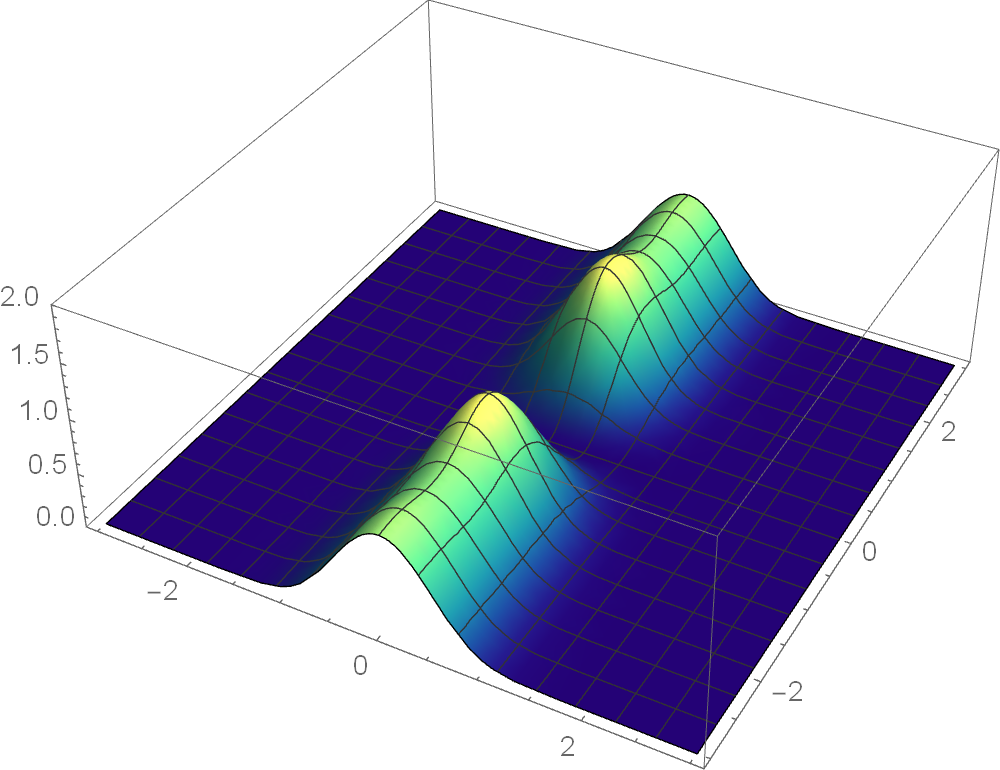}
		\end{center}
		\subcaption{$a=\frac12$}
	\end{subfigure}	
	\begin{subfigure}[h]{0.3\textwidth}
		\begin{center}
			\includegraphics[width=\textwidth]{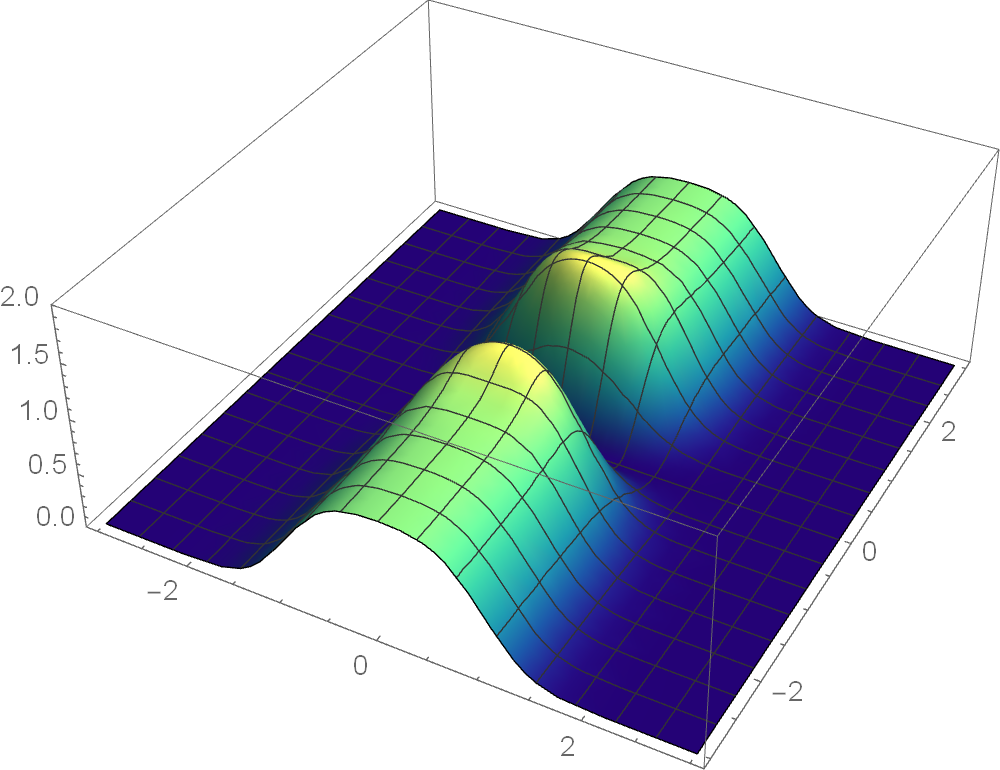}
		\end{center} \subcaption{$a=1$}
	\end{subfigure}
	\begin{subfigure}[h]{0.3\textwidth}
		\begin{center}
			\includegraphics[width=\textwidth]{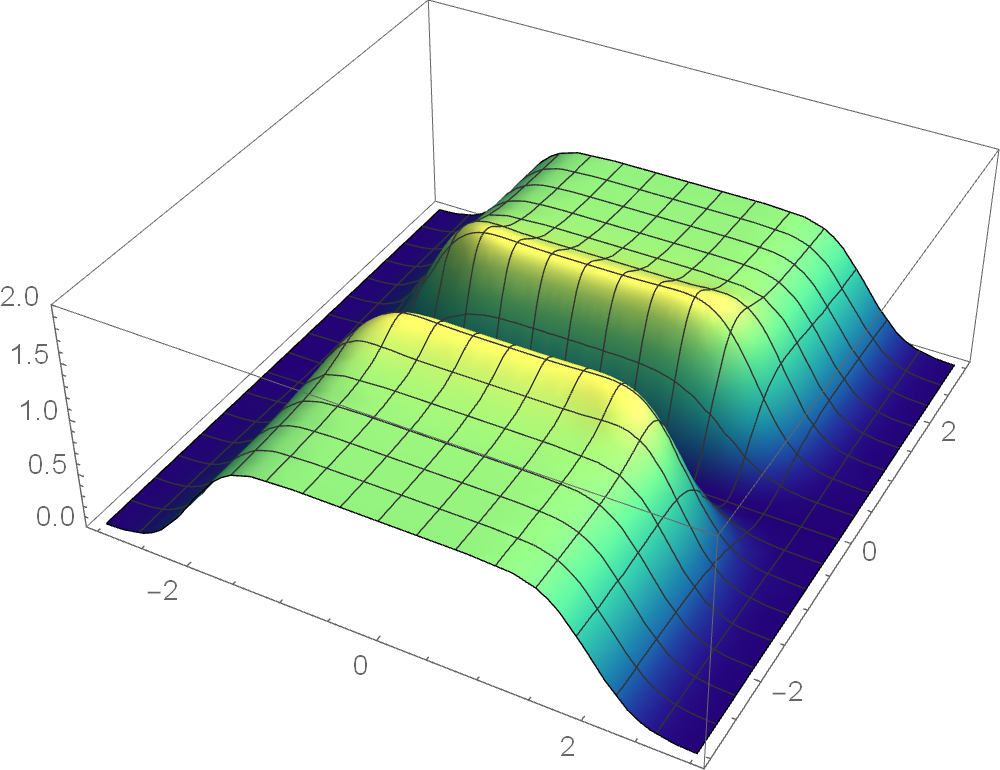}
		\end{center} \subcaption{$a=2$}
	\end{subfigure}
	\caption{ The functions $(x,y) \mapsto R^{\mathbb{C}}_{1} (x+iy)$ (first row) and $(x,y) \mapsto R^{\mathbb{R}}_{1} (x+iy)$ (second row) for the indicated values of $a$. 
	} \label{Fig_R Pfaffian real} \label{Fig_R Pfaffian complex}
\end{figure}

\begin{remark}\label{remark: matching with cplx Ginibre}
The limiting kernel $K^\C$ in Theorem~\ref{thm:main} (a) already appeared in the study of several determinantal processes \cite{MR1634312, fyodorov1997almost, MR1431718, akemann2016universality, MR4030288, AB21, byun2021random}. The way $K^\C$ arises in Theorem~\ref{thm:main} (a) is novel in that \eqref{Gibbs} is Pfaffian.

As is known \cite[Remark 4.(c)]{akemann2016universality}, $K^\C$ interpolates between two other known limiting kernels: as $\rho \to 0$, $K^\C$ converges (after a proper rescaling) to the sine kernel $K^{ \mathrm{sin}}(x,y):=\frac{\sin(\pi(x-y))}{\pi(x-y)}$, while as $\rho \to\infty$ it converges to the limiting kernel in the bulk of the complex Ginibre ensemble, namely $K^{\mathrm{exp}}(z,w) := \exp \big[z \overline{w} - \frac{1}{2}(|z|^{2}+|w|^{2}) \big]$.
\end{remark}

\begin{remark}\label{remark: mathcing with symplectic Ginibre} 
The skew pre-kernel $\kappa^{\mathbb{R}}$ is new to the best of our knowledge. Here we look at some particular limits involving $\kappa^{\mathbb{R}}$ and $R_{k}^{\mathbb{R}}$ when $\rho \to 0$ and $\rho \to \infty$ (although we emphasize that Theorem \ref{thm:main} is valid only for $\rho$ fixed).
\begin{enumerate}
\item \textup{\textbf{Matching with the bulk limit of the symplectic Ginibre ensemble when $\boldsymbol{\rho \to \infty}$.}} \newline Recall that the (classical) symplectic Ginibre ensemble is the planar Pfaffian point process \eqref{Gibbs} with $a_{N}=1$ and $b_{N}=0$. The limiting pre-kernel $\kappa^{W}$ in the bulk of the symplectic Ginibre ensemble is known (see \cite[Remark 2.3 (ii)]{akemann2021scaling}) and given by
\begin{equation} \label{kappa symplectic Ginibre}
\kappa^{W}(z,w):=\sqrt{\pi} e^{z^2+w^2} \int_{-\infty}^\infty W(f_w,f_z)(u)\,du, 
\end{equation}
where $f_z$ is given in \eqref{fz a} and we recall that $W(f,g):=fg'-gf'$ is the Wronskian. Since $\rho$ is proportional to the width of the droplet (see \eqref{radius}), it is natural to expect $\kappa^{\mathbb{R}}$ to tend to $\kappa^{W}$ as $\rho \to + \infty$. A direct analysis shows that this is the case. Indeed, for any fixed $z,w \in \mathbb{C}$, using that $\lim_{a\to + \infty} f_z(-a)=0$ and that $W(f_{w},f_{z})(u)$ has fast decay as $u \to \pm \infty$, we obtain
\begin{align*}
\lim_{\rho \to \infty }\kappa^{\mathbb{R}}(z,w) = \kappa^{W}(z,w).
\end{align*}
\item \textup{\textbf{Matching with the bulk limit of a chiral Gaussian unitary ensemble when $\boldsymbol{\rho \to 0}$.}} \newline
The chiral Gaussian unitary ensemble with parameter $\nu=\frac{1}{2}$ \cite{JacEdward} is the determinantal point process for $N$ points $\boldsymbol{\xi}=\{\xi_{j}\}_{j=1}^{N}$ on $\mathbb{R}_+$ 
whose joint probability distribution $\P_N^{\chi\mathrm{GUE}}$ is defined by 
\begin{equation} \label{Gibbs chiral}
	d\P_N^{\chi\mathrm{GUE}}(\boldsymbol{\xi}) = \frac{1}{N!Z_N^{\chi\mathrm{GUE}}} \prod_{1 \leq j<k\leq N} \abs{\xi_j^2-\xi_k^2}^2 \prod_{j=1}^{N} \xi_j^2\,e^{-N \xi_j^{2} } \,  d\xi_j,
\end{equation}
and whose limiting correlation kernel at the origin is given by 
\begin{equation} \label{K chiral}
K^{\mathrm{\chi sin}}(y_1,y_2):=\frac{ \sin(4(y_1-y_2)) }{ 2(y_1-y_2) }-\frac{\sin(4(y_1+y_2))}{2(y_1+y_2)}, \qquad y_{1},y_{2}\in \mathbb{R}_+,
\end{equation}
see \cite[Section 7.2]{forrester2010log}. We mention that \eqref{Gibbs chiral} is also the ensemble of anti-symmetric Hermitian matrices when their size is odd \cite[Section 13.1]{Mehta}. The kernel \eqref{K chiral} is also equivalent to the well-known Bessel kernel in squared variables with parameter $\nu=\frac12$ (see e.g. \cite[Eq.(7.54)]{forrester2010log}\footnote{There is a typo in \cite[Eq.(7.54)]{forrester2010log}: the $\pm$ sign on the right-hand side should instead be $\mp$.}). When taking $\rho \to 0$, one heuristically expects the correlation functions $R_k^{\mathbb{R}}$ in \eqref{def of RkR} to ``become determinantal" and to be related to $K^{\mathrm{\chi sin}}$. Indeed, by denoting $\zeta_j=1 + i\xi_j$ ($\xi_j \in \R$) and replacing $dA(\zeta_{j})$ by $d\xi_{j}$ in \eqref{Gibbs}, we formally obtain 
\begin{align}\label{lol3}
d\P_N(\boldsymbol{\zeta}) \sim \prod_{1 \leq j<k \leq N} \abs{\xi_j^2-\xi_k^2}^2 \prod_{j=1}^{N} \xi_j^2\,e^{-N Q_{N}(1+i\xi_{j}) } \,  d\xi_j.
\end{align}
The main difference between \eqref{Gibbs chiral} and \eqref{lol3} lies in the potential, but by universality heuristics one expects this difference to be irrelevant for the computation of the limiting kernel. We make this heuristic more precise in Proposition \ref{Cor_Chiral lim} below.
\end{enumerate}
\end{remark}

\begin{prop} \label{Cor_Chiral lim} \textup{} 
Recall that $a= \frac{\rho}{2\sqrt{2}}$ and that $R_{k}^{\mathbb{R}}$ is defined in \eqref{def of RkR}. Let $k \in \mathbb{N}_{>0}$ and $z_{1},\ldots,z_{k}\in \mathbb{C}$ be fixed. As $\rho \to 0$, we have 
\begin{equation} \label{convergence in distribution}
\frac{1}{a^{2k}} R_k^{\mathbb{R}}\Big(\frac{z_1}{a},\dots, \frac{z_k}{a}\Big) \quad \overset{d}{\longrightarrow} \quad \det \Big[ K^{\mathrm{\chi sin}}(y_j,y_l) \Big]_{j,l=1}^k, \qquad (y_k=\im z_k),
\end{equation}
where $K^{\mathrm{\chi sin}}$ is given by \eqref{K chiral}, and where ``$\overset{d}{\longrightarrow}$" in \eqref{convergence in distribution} means that for any bounded and continuous function $f:\C^k \to \R$ with compact support, we have
\begin{align*}
\small \lim_{\rho\to0}\int_{\C^k} f(z_1,\dots,z_k)\frac{1}{a^{2k}} R_k^{\mathbb{R}}\Big(\frac{z_1}{a},\dots, \frac{z_k}{a}\Big) \prod_{j=1}^k \,dA(z_j) = \int_{\R^k} f(y_1,\dots,y_k) \det \Big[ K^{\mathrm{\chi sin}}(y_j,y_l) \Big]_{j,l=1}^k \prod_{j=1}^k \,dy_j.
\end{align*}
\end{prop}

Our results are summarized in Figure~\ref{Fig_various scaling}, where to lighten the notation we use $y_{j}:=\im z_{j}$ and
$$
K^{W}_{2\times 2}(z,w):=e^{-|z|^2-|w|^2}  \begin{pmatrix} 
	\kappa^{W}(z,w) & \kappa^{W}(z,\bar{w})
	\smallskip  
	\\
	\kappa^{W}(\bar{z},w) & \kappa^{W}(\bar{z},\bar{w}) 
\end{pmatrix}, 
$$
where $\kappa^{W}$ is given by \eqref{kappa symplectic Ginibre}.
\begin{figure}[h!]
\begin{center}
	\begin{tikzcd}
  &[-5em]  \substack{\displaystyle \gamma_N^{2k} \,  \bfR_{N,k}\Big(p+p\gamma_{N}z_{1},\dots,p+p\gamma_{N}z_{k}\Big), \; p = \pm e^{i\frac{\sqrt{2}\rho}{N}t}  \\ \mathrm{Pfaffian}} \arrow[d, "\mathrm{Thm\,\ref{thm:main} (b)}"]{d} \arrow[d, swap, "N \to \infty"] &[-5em]
		\\	
\det \Big[ K^{\mathrm{ \chi sin}}(y_j,y_l) \Big]_{j,l=1}^k 	&	\arrow[l, swap, "\rho \to 0"] \arrow[l, "\mathrm{Prop}\,\ref{Cor_Chiral lim}"]   \substack{\displaystyle R_k^{\mathbb{R}}(z_1+i t,\dots,z_k + i t) \vspace{0.25em} \\ \mathrm{Pfaffian}} \arrow[r, "\rho \to \infty"] \arrow[r, swap, "\mathrm{Rmk}\,\ref{remark: mathcing with symplectic Ginibre}(1)"]  \arrow[d, swap, "t \to \infty"]  \arrow[d, "\mathrm{Thm\,\ref{thm:main}(c)}"] 
		&  \Pf \Big[ 
		K^{W}_{2\times 2}(z_j,z_l)
		\Big]_{j,l=1}^k\prod_{j=1}^k  \frac{2y_j}{i}  
		\\
\det \Big[ K^{\mathrm{ sin}}(y_j,y_l) \Big]_{j,l=1}^k 	& 	\arrow[l, swap, "\rho \to 0"] 	\arrow[l, "\mathrm{Rmk}\,\ref{remark: matching with cplx Ginibre}"]  \substack{\displaystyle R_k^{\mathbb{C}}(z_1,\dots,z_k)  \vspace{0.25em} \\
	\mathrm{Determinant}} \arrow[r, "\rho \to \infty"] \arrow[r, swap, "\mathrm{Rmk}\,\ref{remark: matching with cplx Ginibre}"]  &  \det \Big[ K^{\mathrm{exp}}(z_j,z_l) \Big]_{j,l=1}^k 
		\\
& \ar[u, swap, "\mathrm{Thm\,\ref{thm:main} (a)}"]{u} \ar[u, "N\to \infty"]{u} \substack{\displaystyle  \gamma_N^{2k} \,  \bfR_{N,k}\Big(p+p\gamma_{N}z_{1},\dots,p+p\gamma_{N}z_{k}\Big), \; p=e^{i\theta}, \; p \neq -1,1 \\ \mathrm{Pfaffian}} & 
	\end{tikzcd}  
\end{center} \caption{The second column and the second row summarize our main findings. The third row was already known \cite{MR1634312, akemann2016universality} and is included in this diagram to place our results in their overall context.} \label{Fig_various scaling}
\end{figure}

\subsection{Gap probabilities}

Deriving asymptotic formulas of gap probabilities is a classical problem in random matrix theory with a rich history \cite{F2014, GN2018}. The problem is usually considered challenging when the hole region is large. For a point process with $N$ points, a ``large" hole region refers to a region that contains, with high probability, a number of points proportional to $N$. For the disk hole region in the complex Ginibre point process, this problem was investigated by several authors \cite{GHS1988, MR1181356, MR1239571}. In recent years, these results have been extended to other potentials, hole regions and models in \cite{MR2536111, MR3063493, MR3279619, MR3719476, MR3814242, L2019}, and have been improved to higher precision in \cite{charlier2021large}. For more results on large gap asymptotics of two-dimensional point processes, we refer to \cite{GN2018, charlier2021large}. To our knowledge, the only paper prior this work on large gap probabilities of planar Pfaffian point processes is \cite{MR2536111}. 

\medskip In this subsection, we obtain large $N$ asymptotics, up to and including the term of order $1$, for the following three gap probabilities:
\begin{align}
& \Prob_N^1 :=\mathbb{P}\Big(\# \Big\{ \zeta_j: |\zeta_j| \in [0,r_1] \Big\} =0  \Big)=\mathbb{P}\Big(\# \Big\{ \zeta_j: |\zeta_j| \in [r_1,\infty] \Big\} =N  \Big),  \label{PN 1}
\\
& \mathbb{P}_N^2  :=\mathbb{P}\Big(\# \Big\{ \zeta_j: |\zeta_j| \in [r_2,\infty] \Big\} =0  \Big)=\mathbb{P}\Big(\# \Big\{ \zeta_j: |\zeta_j| \in [0,r_2] \Big\} =N  \Big), \label{PN 2}
\\
& \mathbb{P}_N^{12}  :=\mathbb{P}\Big(\# \Big\{ \zeta_j: |\zeta_j| \in [0,r_1] \cup [r_2,\infty] \Big\} =0  \Big)=\mathbb{P}\Big(\# \Big\{ \zeta_j: |\zeta_j| \in [r_1,r_2] \Big\} =N  \Big), \label{PN 12}
\end{align}
where the $\zeta_{j}$ are distributed according to \eqref{Gibbs} with $Q_N$ as in \eqref{Q aN bN}. Exact expressions for these probabilities in terms of the incomplete gamma function are given in \eqref{PN1 gamma}, \eqref{PN2 gamma} and \eqref{PN12 gamma} respectively. 

\medskip As argued in \ref{item gap}, for typical configurations the hole regions associated with $\Prob_N^1$, $\Prob_N^2$ and $\Prob_N^{12}$ contain about $\sim \sqrt{N}$ points. Hence, the problems of determining the large $N$ asymptotics of $\Prob_N^1$, $\Prob_N^2$, $\Prob_N^{12}$ can be seen as ``semi-large" gap problems (and are simpler than ``large" gap problems such as the ones considered in \cite{MR2536111}). 

\newpage 
\begin{thm} \label{Thm_partition functions}\textup{\textbf{Asymptotics of semi-large gap probabilities}} 

\noindent Let $\rho \in (0,\infty)$ be fixed. As $N \to \infty$, we have
\begin{align}
& \Prob_N^1 =  \exp\Big( N C_1 - C_0 + O(N^{-1}) \Big), 
\qquad 
\Prob_N^2= \exp\Big( N C_1 + C_0 + O(N^{-1}) \Big),      \label{PN 1 and 2 asym thm}
\\
& \Prob_N^{12} =  \exp\Big( N \wt{C}_1 + O(N^{-1}) \Big), \label{PN 12 asym thm}
\end{align}
where 
\begin{align}
& C_1=\int_{0}^{1} \log \Big( 1-\frac{1}{2}\erfc(\sqrt{2}\, \rho x) \Big)\,dx,  \label{C1(rho)}
\\
& \wt{C}_1=\int_{0}^{1} \log \Big( \frac{1}{2}\mathrm{erfc}( \sqrt{2}\, \rho (x-1))-\frac{1}{2}\mathrm{erfc}(\sqrt{2}\, \rho x) \Big)\,dx,   \label{C1(rho) wt} \\
& C_0= \frac{\log \big( 2-\mathrm{erfc}(\sqrt{2}\rho) \big)}{2} - \frac{\rho }{3\sqrt{2\pi}} \int_{0}^{1}\frac{ e^{-2\rho^{2}x^{2}}\,(5+3\rho^{2} x - 2\rho^{2}x^{2})}{ 1-\frac{1}{2}\mathrm{erfc}(\sqrt{2}\, \rho x)}\,dx. \label{C0(rho)}
\end{align}
\end{thm}
Theorem~\ref{Thm_partition functions} has been verified numerically, see Figure~\ref{Fig_semi large gap}. 

\begin{figure}[h!]
	
		\begin{subfigure}{0.32\textwidth}
		\begin{center}	
			\includegraphics[width=\textwidth]{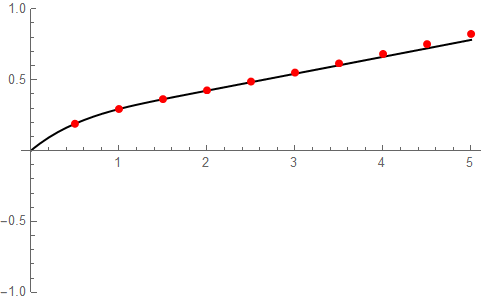}
		\end{center}
		\subcaption{\small $-C_{0}$ and $\log \Prob_N^1-NC_1$}
	\end{subfigure}	
	\begin{subfigure}[h]{0.32\textwidth}
		\begin{center}
			\includegraphics[width=\textwidth]{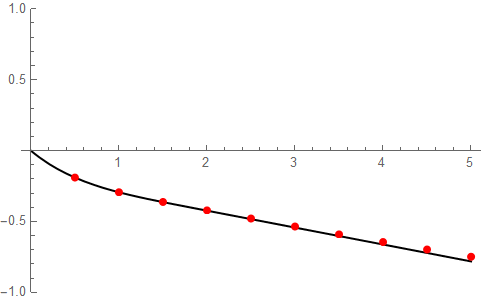}
		\end{center} \subcaption{\small $C_{0}$ and $\log \Prob_N^{2}-NC_1$}
	\end{subfigure}
		\begin{subfigure}[h]{0.32\textwidth}
		\begin{center}
			\includegraphics[width=\textwidth]{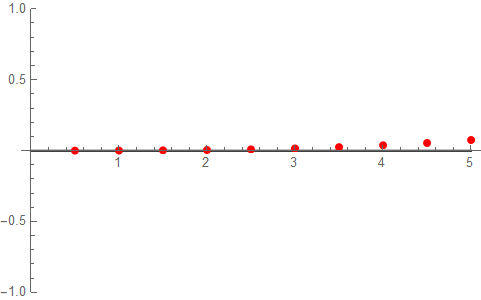}
		\end{center} \subcaption{\small $0$ and $\log \Prob_N^{12}-N\wt{C}_1$}
	\end{subfigure}
	\caption{ The functions $\rho \mapsto -C_0, C_0, 0$ (black lines) versus $\log \Prob_N^1-NC_1$, $\log \Prob_N^2-NC_1$ and $\log \Prob_N^{12}-N\wt{C}_1$ with $N=100$ and various values of $\rho$ (red dots). 
	} \label{Fig_semi large gap} 
\end{figure}

\begin{remark} \textup{ \textbf{Inequalities among $\Prob_N^1$, $\Prob_N^2$, and $\Prob_N^{12}$ for large $N$} }

\noindent It readily follows from the definitions \eqref{PN 1}, \eqref{PN 2} and \eqref{PN 12} that $\Prob_N^1, \Prob_N^2 > \Prob_N^{12}$. Theorem \ref{Thm_partition functions} allows to obtain more quantitative comparisons between these three probabilities. Indeed, it follows from \eqref{C1(rho)}, \eqref{C1(rho) wt} and  $\erfc(x)<2$ ($x \in \R$) that $C_1 > \wt{C}_1$ for all $\rho \in (0,\infty)$.
Moreover, by \eqref{C0(rho)}, $C_0<0$ for all $\rho \in (0,\infty)$ (see also Figure \ref{Fig_semi large gap} (B)).
Thus we have 
\begin{equation} \label{PN 1 2 12 ineq}
\log \Prob_N^1 \underbrace{>}_{O(1)} \log \Prob_N^2 \underbrace{\gg}_{O(N)} \log \Prob_N^{12}, \qquad \mbox{as } N \to \infty.  
\end{equation}
In \eqref{PN 1 2 12 ineq}, the notation below $>$ and $\gg$ indicates the order of the difference between the left- and right-hand sides. 
\end{remark}

%
%

\subsection{Outline} 

Different strategies will be used to prove Theorem \ref{thm:main} (a) and (b). However, for both, the starting point is the same: using skew-orthogonal polynomial techniques, we first express $\gamma_N^{2k} \,  \bfR_{N,k}\big(p+p\gamma_{N}z_{1},\dots,p+p\gamma_{N}z_{k}\big)$ and the associated normalized skew pre-kernel $\wh{\bfkappa}_N$ in terms of the incomplete Gamma function (see \eqref{bfkappa wh},\eqref{incom Gamma} and Lemma \ref{Lem_bfkappa exp}). We then prove a so-called generalised Christoffel-Darboux formula for $\wt{\bfkappa}_N(\zeta,\eta):=e^{-a_NN \zeta \eta}\wh{\bfkappa}_N(\zeta,\eta)$, which in our case is a differential identity expressing $\pa_\zeta \wt{\bfkappa}_N$ in terms of $\wt{\bfkappa}_N$ and the incomplete Gamma function (see Proposition~\ref{Prop_CDI}).

In Section~\ref{Section_Scaling limits} we prove Theorem \ref{thm:main}. Parts (a), (b) and (c) of this theorem are proven in Subsections~\ref{Subsec_cplex limit}, \ref{Subsec_real limit} and \ref{Subsec_scaling interpolation}, respectively.  In short, part (a) is proved via a first order Riemann sum approximation, part (b) is proved using our generalised Christoffel-Darboux formula, while part (c) follows from a direct analysis of $R_k^{\mathbb{R}}$. Proposition~\ref{Cor_Chiral lim} is proved in Subsection~\ref{Subsec_chiral limit}.

Theorem~\ref{Thm_partition functions} is proved in Section~\ref{Section_Parition functions}. Using again skew-orthogonal polynomials, we express of $\mathbb{P}_{N}^{1},\mathbb{P}_{N}^{2},\mathbb{P}_{N}^{12}$ in terms of the incomplete gamma function $\gamma$. We then use some known asymptotics expansion of $\gamma$ (which we recall in Appendix~\ref{appendix_gamma}) together with some precise Riemann sum approximations to complete the proof of Theorem~\ref{Thm_partition functions}.

\subsection{Notation.} 
Throughout this paper, $p=e^{i\theta}$ denotes the base/zooming point on the unit circle around which the local statistics of \eqref{Gibbs} are considered. Given $\boldsymbol{z}=\{ z_j \}_{j=1}^N$, we define $\bfs{\zeta}= \{ \zeta_j \}_{j=1}^N \subset \mathbb{C}$ by
\begin{equation} \label{rescaling}
\zeta_j := p+ p \gamma_N z_j, \qquad \gamma_N=\sqrt{ \frac{2}{a_N N} }, \qquad j=1,\ldots,N,
\end{equation}
and to shorten the notation we also define
\begin{equation} \label{RNk bfRNk}
	R_{N,k}(z_1,\dots, z_k):= \gamma_N^{2k} \,  \bfR_{N,k}\big(p+p\gamma_{N}z_{1},\dots,p+p\gamma_{N}z_{k}\big) = \gamma_N^{2k} \,  \bfR_{N,k}(\zeta_1,\dots,\zeta_k),
\end{equation}
where we recall that $\bfR_{N,k}$ is given by \eqref{bfRNk def}.

\section{Skew-orthogonal polynomials and Christoffel-Darboux formula}\label{Section_Preliminaries}

In this section we first use the skew-orthogonal formalism for planar symplectic ensembles introduced by Kanzieper \cite{MR1928853} to express $\bfkappa_N$ as in \eqref{bfkappaN GN}--\eqref{GN} in terms of $\Gamma$. We then obtain a generalised Christoffel-Darboux formula for a normalized skew pre-kernel (Proposition~\ref{Prop_CDI}).

\medskip \noindent \textbf{Skew-orthogonal polynomials.} Consider the following skew-symmetric form $\langle \cdot , \cdot  \rangle_s$
\begin{equation*}
\langle f, g \rangle_s := \int_{\C} \Big( f(\zeta) g(\bar{\zeta}) - g(\zeta) f(\bar{\zeta}) \Big) (\zeta - \bar{\zeta}) e^{-N Q_{N}(\zeta)} \,dA(\zeta).
\end{equation*}
A family $\{q_{m}\}_{m \geq 0}$ of monic polynomials $q_m$ of degree $m$ is said to be a family of skew-orthogonal polynomials if the following skew-orthogonality conditions hold: for all $k, l \in \mathbb{N}$
\begin{equation}\label{lol1}
\langle q_{2k}, q_{2l} \rangle_s = \langle q_{2k+1}, q_{2l+1} \rangle_s = 0, \qquad \langle q_{2k}, q_{2l+1} \rangle_s = -\langle q_{2l+1}, q_{2k} \rangle_s = r_k  \,\delta_{k, l},
\end{equation}
where $r_k$ is a positive constant called the $k$-th skew-norm and $\delta_{k, l}$ is the Kronecker delta. The existence of $\{q_{m}\}_{m \geq 0}$ follows from a Gram-Schmidt skew-orthogonalisation procedure \cite[Theorem 2.4]{akemann2021skew}, and a sufficient condition to ensure uniqueness is to set the coefficient $z^{2k}$ in $q_{2k+1}(z)$ to $0$ for each $k \in \mathbb{N}$ \cite[Lemma 2.2]{akemann2021skew}. It follows from the general theory \cite{MR1928853} that \eqref{bfR Pfa} holds with
\begin{equation}\label{bfkappaN skewOP}
\bfkappa_N(\zeta,\eta)=\sum_{k=0}^{N-1} \frac{q_{2k+1}(\zeta) q_{2k}(\eta) -q_{2k}(\zeta) q_{2k+1}(\eta)}{r_k}.
\end{equation}
(In fact the above theory from \cite{MR1928853} applies in a much broader setting, for example for smooth potentials $Q:\mathbb{C}\to \mathbb{R}$ of sufficient increase near $\infty$ (not necessarily rotation-invariant), but this will not be needed for us.)

\begin{remark}\label{remark:non-uniqueness of the kernel}
Different pre-kernels can yield the same correlation functions $\bfR_{N,k}$.
For instance, if $\{g_N:\C \to \C\}_{N=0}^{+\infty}$ is a sequence satisfying $g_N(\overline{\zeta}) = 1/g_N(\zeta)$ for each $\zeta \in \mathbb{C}$ and $N \in \mathbb{N}$, then replacing $\bfkappa_N(\zeta, \eta)$ in \eqref{bfR Pfa} by $g_N(\zeta) g_N(\eta) \bfkappa_N(\zeta, \eta)$ does not modify $\bfR_{N,k}$; in other words, the two pre-kernels $\bfkappa_N(\zeta, \eta)$ and $g_N(\zeta) g_N(\eta) \bfkappa_N(\zeta, \eta)$ give rise to same point process.
\end{remark} 

Since $Q_{N}$ is rotation-invariant, by e.g.~\cite[p.7]{MR3066113} and \cite[Corollary 3.3]{akemann2021skew} the following holds: the family $\{q_{k}\}_{k=0}^{+\infty}$ defined by
\begin{equation} \label{skew op_rad}
	q_{2k+1}(\zeta)=\zeta^{2k+1}, \qquad
	q_{2k}(\zeta)=\zeta^{2k}+\sum_{l=0}^{k-1}  \zeta^{2l} \prod_{j=0}^{k-l-1} \frac{h_{2l+2j+2}  }{ h_{2l+2j+1} }, \qquad k=0,1,\ldots
\end{equation}
where $h_{k}$ is given by
\begin{equation} \label{hk Qradial}
    h_k:=\int_\C |\zeta|^{2k} e^{-N Q_{N}(\zeta)} \, dA(\zeta),
\end{equation}
satisfies \eqref{lol1} with
\begin{equation}\label{rk h2k+1}
r_k=2\,h_{2k+1}.
\end{equation}

For definiteness, from now we let $\{q_{k}\}_{k=0}^{+\infty}$ be the family of skew-orthogonal polynomials defined in \eqref{skew op_rad}, and we define $\bfkappa_N$ as in \eqref{bfkappaN skewOP} in terms of those $\{q_{k}\}_{k=0}^{+\infty}$. It follows from \eqref{bfR Pfa}, the rescaling \eqref{RNk bfRNk} and the definition \eqref{Q aN bN} of $Q_{N}$ that 
\begin{equation}
	\begin{split} \label{Rk bfkappa wh}
		R_{N,k}(z_1,\dots, z_k) 
		&= \Pf \Big[ e^{ -\frac{a_NN}{2}(|\zeta_j|^2+|\zeta_l|^2) }   \begin{pmatrix} 
			\wh{\bfkappa}_N(\zeta_j,\zeta_l) & \wh{\bfkappa}_N(\zeta_j,\bar{\zeta}_l)
			\smallskip 
			\\
			\wh{\bfkappa}_N(\bar{\zeta}_j,\zeta_l) & \wh{\bfkappa}_N(\bar{\zeta}_j,\bar{\zeta}_l) 
		\end{pmatrix}   \Big]_{ j,l=1 }^k \prod_{j=1}^{k} \frac{\bar{\zeta}_j-\zeta_j}{\gamma_N}  ,
	\end{split}
\end{equation}
where 
\begin{equation} \label{bfkappa wh}
	\wh{\bfkappa}_N(\zeta,\eta)=\gamma_N^3\, \zeta^{b_NN} \eta^{b_NN} \bfkappa_N(\zeta,\eta), 
\end{equation}
and the principal branches are used for $\zeta^{b_NN}$ and $\eta^{b_NN}$. To obtain \eqref{bfkappa wh}, we have also used Remark \ref{remark:non-uniqueness of the kernel} with $g_{N}(\zeta) = \zeta^{b_NN}/|\zeta|^{b_NN}$ to replace $|\zeta \eta|^{b_NN}$ in the calculations by $\zeta^{b_NN} \eta^{b_NN}$. 

Recall that the incomplete gamma functions $\Gamma(a,\zeta)$, $\gamma(a,\zeta)$ and the regularised gamma function $\mathrm{Q}(a,\zeta)$ are given by 
\begin{equation} \label{incom Gamma}
	\mathrm{Q}(a,\zeta):= \frac{\Gamma(a,\zeta)}{\Gamma(a)}, \quad \Gamma(a,\zeta):=\int_{\zeta}^\infty t^{a-1}e^{-t}\,dt, \quad \gamma(a,\zeta):=\int_{0}^\zeta t^{a-1}e^{-t}\,dt = \Gamma(a)-\Gamma(a,\zeta),
\end{equation}
see also e.g. \cite[Chapter 8]{olver2010nist}.

\begin{lem} \label{Lem_bfkappa exp}
Let $\zeta,\eta \in \mathbb{C}$. We have 
$\wh{\bfkappa}_N(\zeta,\eta)=\wh{\boldsymbol{G}}_N(\zeta,\eta)-\wh{\boldsymbol{G}}_N(\eta,\zeta),$
where 
\begin{align}
& \wh{\boldsymbol{G}}_N(\zeta,\eta):= \sqrt{\pi} \sum_{k=0}^{N-1} \frac{ ( \sqrt{\frac{a_NN}{2}}\, \zeta )^{2k+1+b_NN} }{  \Gamma(k+\tfrac{3}{2}+\tfrac{b_N N}{2}) }  \sum_{l=0}^k \frac{ ( \sqrt{ \frac{a_NN}{2} } \eta )^{2l+b_NN} }{ \Gamma(l+\frac{b_N N }{2}+1) } \nonumber 
\\
& = \sqrt{\pi} \sum_{k=0}^{N-1} \frac{ ( \sqrt{\frac{a_NN}{2}}\, \zeta )^{2k+1+b_NN} }{  \Gamma(k+\tfrac{3}{2}+\tfrac{b_N N}{2}) }  e^{ \frac{a_NN}{2} \eta^2 } \bigg( \mathrm{Q}\Big(k+\tfrac{b_NN}{2}+1, \tfrac{a_NN}{2} \eta^2\Big) -\mathrm{Q}\Big(\tfrac{b_NN}{2}, \tfrac{a_NN}{2} \eta^2\Big)   \bigg).    \label{G bold hat}
\end{align}
\end{lem}
\begin{proof}

Recall that $Q_{N}$ is defined in \eqref{Q aN bN}. Hence, by \eqref{hk Qradial} we have
\begin{equation} \label{hk induced}
h_k= 2\int_{0}^\infty r^{2k+1} e^{-N Q(r)}\,dr = 2 \int_0^\infty r^{2k+1+2b_N N} e^{-a_N N r^2} \,dr =  \frac{ \Gamma(1+k+b_N N)  } { (a_N N)^{1+k+b_N N}  },
\end{equation}
and thus
\begin{align*}
\prod_{j=0}^{k-l-1} \frac{h_{2l+2j+2}  }{ h_{2l+2j+1} } & = \prod_{j=0}^{k-l-1} \frac{ \Gamma(2l+2j+3+b_N N)  } { (a_N N)^{2l+2j+3+b_N N}  }   \frac { (a_N N)^{2l+2j+2+b_N N}  } { \Gamma(2l+2j+2+b_N N)  } 
\\
&=  \prod_{j=0}^{k-l-1}  \frac{ l+j+1+b_N N/2 }{ a_N N/2 } 
= \Big( \frac{2}{a_N N} \Big)^{k-l} \frac{ \Gamma(k+\frac{b_N N}{2}+1) }{ \Gamma(l+\frac{b_N N }{2}+1) }, \quad 0 \leq l < k.
\end{align*}
Substituting the above in the definition \eqref{skew op_rad} of $q_{2k}$ yields
\begin{equation} \label{skew op_induced}
q_{2k}(\zeta)= \sum_{l=0}^k \Big( \frac{2}{a_N N} \Big)^{k-l} \frac{ \Gamma(k+\frac{b_N N}{2}+1) }{ \Gamma(l+\frac{b_N N }{2}+1) } \zeta^{2l}, \qquad k \geq 0.
\end{equation}
Using \eqref{rk h2k+1}, \eqref{hk induced}, and the duplication formula of the gamma function
\begin{equation}\label{Gamma duplication}
 \Gamma(2z)=\frac{2^{2z-1}}{\sqrt{\pi}} \Gamma(z)\Gamma(z+\tfrac12),
\end{equation}
we obtain
\begin{equation} \label{rk induced}
r_k=2\,\frac{ \Gamma(2+2k+b_N N)  } { (a_N N)^{2+2k+b_N N}  }= \frac{1}{\sqrt{\pi}} \Big( \frac{2}{a_NN} \Big)^{2+2k+b_NN} \Gamma(k+\tfrac{b_NN}{2}+1) \Gamma(k+\tfrac{b_NN}{2}+\tfrac{3}{2}). 
\end{equation}
Now, by \eqref{bfkappaN skewOP}, we get \eqref{bfkappaN GN} with 
\begin{equation}
\begin{split}
\boldsymbol{G}_N(\zeta,\eta)=\sum_{k=0}^{N-1} \frac{q_{2k+1}(\zeta) q_{2k}(\eta) }{r_k}
&=\sum_{k=0}^{N-1} \frac{ (a_N N)^{2k+2+b_N N} }{ 2\,\Gamma(2k+2+b_N N) } \zeta^{2k+1} \sum_{l=0}^k \Big( \frac{2}{a_N N} \Big)^{k-l} \frac{ \Gamma(k+\frac{b_N N}{2}+1) }{ \Gamma(l+\frac{b_N N }{2}+1) } \eta^{2l}
\\
&=  \sqrt{\pi} \Big( \frac{a_NN}{2} \Big)^{b_NN+\frac32}  \sum_{k=0}^{N-1} \frac{ ( \sqrt{\frac{a_NN}{2}}\, \zeta )^{2k+1} }{  \Gamma(k+\tfrac{3}{2}+\tfrac{b_N N}{2}) }  \sum_{l=0}^k \frac{ ( \sqrt{ \frac{a_NN}{2} } \eta )^{2l} }{ \Gamma(l+\frac{b_N N }{2}+1) }.
\end{split}
\end{equation}
It then follows from \eqref{bfkappa wh} that $\wh{\bfkappa}_N(\zeta,\eta)=\wh{\boldsymbol{G}}_N(\zeta,\eta)-\wh{\boldsymbol{G}}_N(\eta,\zeta),$ where 
$\wh{\boldsymbol{G}}_N(\zeta,\eta):= \gamma_N^3\zeta^{b_NN} \eta^{b_NN}\boldsymbol{G}_N(\zeta,\eta). $ This proves the first equation in \eqref{G bold hat}. The second expression in \eqref{G bold hat} follows from the recurrence relation of the incomplete Gamma function (see e.g. \cite[Eq.(8.8.9)]{olver2010nist}), namely
\begin{equation}\label{sum to Q}
	\sum_{k=0}^{N-1} \frac{ z^{k+c} }{ \Gamma(k+c+1) } = e^z \Big( \mathrm{Q}(N+c,z)-\mathrm{Q}(c,z) \Big).
\end{equation}
\end{proof}

We now obtain an identity for $\pa_\zeta \wt{\bfkappa}_N(\zeta,\eta)$ in terms of $\wt{\bfkappa}_N(\zeta,\eta)$ and $\mathrm{Q}$. As mentioned earlier, such identities are typically called generalised Christoffel-Darboux formulas (see e.g. \cite[Proposition 2.3]{MR3450566}).

\begin{prop}(Christoffel-Darboux formula) \label{Prop_CDI}
For $\theta \in [0,2\pi),$ and $\zeta,\eta \in \mathbb{C},$ let 
\begin{equation} \label{bfkappa wt}
\wt{\bfkappa}_N(\zeta,\eta):=e^{-a_NN \zeta \eta}\wh{\bfkappa}_N(\zeta,\eta)=e^{-2\mu \nu }\wh{\bfkappa}_N(\zeta,\eta) ,
\end{equation}
where 
\begin{equation} \label{mu nu}
\mu:=\sqrt{ \frac{a_NN}{2} } \zeta, \qquad \nu:=\sqrt{ \frac{a_NN}{2} } \eta.
\end{equation}
Then we have 
\begin{align}
\sqrt{\frac{2}{a_NN}} \pa_\zeta \wt{\bfkappa}_N(\zeta,\eta)& = 2(\mu-\nu) \, \wt{\bfkappa}_N(\zeta,\eta) + 2\, \Big( \mathrm{Q}(2N+b_NN,2\mu \nu)-\mathrm{Q}(b_NN, 2\mu \nu) \Big) \nonumber
\\
&\quad  -2\,  \sqrt{\pi}\, \frac{ \mu^{2N+b_NN} }{  \Gamma(N+\tfrac{1}{2}+\tfrac{b_N N}{2}) }  \,e^{\nu^2-2\mu \nu} \Big( \mathrm{Q}(N+\tfrac{b_NN}{2}, \nu^2)-\mathrm{Q}(\tfrac{b_NN}{2},\nu^2) \Big) \nonumber
\\
&\quad -2\, \sqrt{\pi} \, \frac{ \mu^{b_NN-1} }{ \Gamma(\frac{b_N N }{2}) }  \,e^{\nu^2-2\mu \nu} \Big( \mathrm{Q}(N+\tfrac{b_NN+1}{2}, \nu^2)-\mathrm{Q}(\tfrac{b_NN+1}{2},\nu^2) \Big). \label{CDI bfkappa}
\end{align}
\end{prop}

\begin{remark} \label{remark:CDI}
Curiously, the term $\mathrm{Q}(2N+b_NN,2\mu \nu)-\mathrm{Q}(b_NN, 2\mu \nu)$ appearing in the right-hand side of \eqref{CDI bfkappa} is equal to the kernel of the complex induced Ginibre ensemble times $e^{\frac{N}{2}(Q_{N}(\nu)+Q_{N}(\mu))}$. 
Similar identities featuring possible relations between orthogonal and skew-orthogonal polynomial kernels have already appeared in the literature, see \cite{byun2021universal, akemann2021scaling, MR2180006, MR3612266} for two-dimensional point processes and \cite{MR1762659, MR1675356} for one-dimensional point processes.  
\end{remark}

\begin{proof}[Proof of Proposition~\ref{Prop_CDI}]
By \eqref{G bold hat}, we have 
\begin{equation}
\wh{\boldsymbol{G}}_N(\zeta,\eta):= \sqrt{\pi} \sum_{k=0}^{N-1} \frac{ \mu^{2k+1+b_NN} }{  \Gamma(k+\tfrac{3}{2}+\tfrac{b_N N}{2}) }  \sum_{l=0}^k \frac{ \nu^{2l+b_NN} }{ \Gamma(l+\frac{b_N N }{2}+1) },
\end{equation}
where $\mu$ and $\nu$ are given by \eqref{mu nu}. 
Differentiating $\wh{\boldsymbol{G}}_N(\zeta,\eta)$ with respect to the $\zeta$-variable yields
\begin{equation}
\begin{split}
\pa_\zeta \wh{\boldsymbol{G}}_N(\zeta,\eta) &=   \sqrt{ \pi \frac{a_NN}{2} }  \sum_{k=0}^{N-1} \frac{ (2k+1+b_NN) \mu^{2k+b_NN} }{  \Gamma(k+\tfrac{3}{2}+\tfrac{b_N N}{2}) }  \sum_{l=0}^k \frac{ \nu^{2l+b_NN} }{ \Gamma(l+\frac{b_N N }{2}+1) } 
\\
& = \sqrt{\pi 2a_NN }  \sum_{k=0}^{N-1} \frac{ \mu^{2k+b_NN} }{  \Gamma(k+\tfrac{1}{2}+\tfrac{b_N N}{2}) }  \sum_{l=0}^k \frac{ \nu^{2l+b_NN} }{ \Gamma(l+\frac{b_N N }{2}+1) }. 
\end{split}
\end{equation}
Rearranging the summation, we have 
\begin{equation}
\begin{split}
\pa_\zeta \wh{\boldsymbol{G}}_N(\zeta,\eta) & =  \sqrt{ \pi 2a_NN } \Bigg[ \sum_{k=1}^{N-1} \frac{ \mu^{2k+b_NN} }{  \Gamma(k+\tfrac{1}{2}+\tfrac{b_N N}{2}) }  \sum_{l=0}^k \frac{ \nu^{2l+b_NN} }{ \Gamma(l+\frac{b_N N }{2}+1) }
+  \frac{ \mu^{b_NN} }{  \Gamma(\tfrac{1}{2}+\tfrac{b_N N}{2}) }  \frac{ \nu^{b_NN} }{ \Gamma(\frac{b_N N }{2}+1) } \Bigg]. 
\end{split}
\end{equation}
Since 
\begin{align*}
& \sum_{k=1}^{N-1} \frac{ \mu^{2k+b_NN} }{  \Gamma(k+\tfrac{1}{2}+\tfrac{b_N N}{2}) }  \sum_{l=0}^k \frac{ \nu^{2l+b_NN} }{ \Gamma(l+\frac{b_N N }{2}+1) } = \mu \sum_{k=0}^{N-2} \frac{ \mu^{2k+1+b_NN} }{  \Gamma(k+\tfrac{3}{2}+\tfrac{b_N N}{2}) }  \sum_{l=0}^{k+1} \frac{ \nu^{2l+b_NN} }{ \Gamma(l+\frac{b_N N }{2}+1) }   
\\
&= \mu \sum_{k=0}^{N-2} \frac{ \mu^{2k+1+b_NN} }{  \Gamma(k+\tfrac{3}{2}+\tfrac{b_N N}{2}) }  \sum_{l=0}^{k} \frac{ \nu^{2l+b_NN} }{ \Gamma(l+\frac{b_N N }{2}+1) }  + \sum_{k=1}^{N-1} \frac{ \mu^{2k+b_NN} }{  \Gamma(k+\tfrac{1}{2}+\tfrac{b_N N}{2}) }  \frac{ \nu^{2k+b_NN} }{ \Gamma(k+\frac{b_N N }{2}+1) } ,
\end{align*}
we have 
\begin{align*}
& \sqrt{\pi} \Bigg[ \sum_{k=1}^{N-1} \frac{ \mu^{2k+b_NN} }{  \Gamma(k+\tfrac{1}{2}+\tfrac{b_N N}{2}) }  \sum_{l=0}^k \frac{ \nu^{2l+b_NN} }{ \Gamma(l+\frac{b_N N }{2}+1) }
+  \frac{ \mu^{b_NN} }{  \Gamma(\tfrac{1}{2}+\tfrac{b_N N}{2}) }  \frac{ \nu^{b_NN} }{ \Gamma(\frac{b_N N }{2}+1) } \Bigg]
\\
&= \mu\,\sqrt{\pi} \sum_{k=0}^{N-2} \frac{ \mu^{2k+1+b_NN} }{  \Gamma(k+\tfrac{3}{2}+\tfrac{b_N N}{2}) }  \sum_{l=0}^{k} \frac{ \nu^{2l+b_NN} }{ \Gamma(l+\frac{b_N N }{2}+1) }  +\sqrt{\pi} \sum_{k=0}^{N-1} \frac{ \mu^{2k+b_NN} }{  \Gamma(k+\tfrac{1}{2}+\tfrac{b_N N}{2}) }  \frac{ \nu^{2k+b_NN} }{ \Gamma(k+\frac{b_N N }{2}+1) }
\\
&= \mu \, \wh{\boldsymbol{G}}_N(\zeta,\eta) -  \sqrt{\pi} \frac{ \mu^{2N+b_NN} }{  \Gamma(N+\tfrac{1}{2}+\tfrac{b_N N}{2}) }  \sum_{l=0}^{N-1} \frac{ \nu^{2l+b_NN} }{ \Gamma(l+\frac{b_N N }{2}+1) }  + \sum_{k=0}^{N-1} \frac{(2\mu \nu)^{2k+b_NN} }{  \Gamma(2k+1+b_N N) },
\end{align*}
where we have used \eqref{Gamma duplication} for the last line. 
We have just shown that 
\begin{equation}
\frac{ \pa_\zeta \wh{\boldsymbol{G}}_N(\zeta,\eta)  }{ \sqrt{2a_NN} }= \mu \, \wh{\boldsymbol{G}}_N(\zeta,\eta) -  \sqrt{\pi} \frac{ \mu^{2N+b_NN} }{  \Gamma(N+\tfrac{1}{2}+\tfrac{b_N N}{2}) }  \sum_{l=0}^{N-1} \frac{ \nu^{2l+b_NN} }{ \Gamma(l+\frac{b_N N }{2}+1) }  + \sum_{k=0}^{N-1} \frac{(2\mu)^{2k+b_NN}\nu^{2k+b_NN} }{  \Gamma(2k+1+b_N N) }.
\end{equation}
In a similar way, we obtain
\begin{equation}
\begin{split}
\frac{\pa_\zeta\wh{\boldsymbol{G}}_N(\eta,\zeta)}{\sqrt{2a_NN}} &= \mu \,\wh{\boldsymbol{G}}_N(\eta,\zeta)- \sum_{k=0}^{N-1} \frac{ (2\mu)^{2k+1+b_NN} \nu^{2k+1+b_NN} }{ \Gamma(2k+2+b_N N) }+\sqrt{\pi} \frac{ \mu^{b_NN-1} }{ \Gamma(\frac{b_N N }{2}) }  \sum_{k=0}^{N-1} \frac{ \nu^{2k+1+b_NN} }{  \Gamma(k+\tfrac{3}{2}+\tfrac{b_N N}{2}) } .
\end{split}
\end{equation}
Combining above equations, we have 
\begin{equation}
\begin{split}
\frac{ \pa_\zeta \wh{\bfkappa}_N(\zeta,\eta)  }{ \sqrt{2a_NN} }&= \mu \, \wh{\bfkappa}_N(\zeta,\eta)  + \sum_{k=0}^{2N-1} \frac{(2\mu)^{k+b_NN} \nu^{k+b_NN} }{  \Gamma(k+1+b_N N) }
\\
& -  \sqrt{\pi} \frac{ \mu^{2N+b_NN} }{  \Gamma(N+\tfrac{1}{2}+\tfrac{b_N N}{2}) }  \sum_{l=0}^{N-1} \frac{ \nu^{2l+b_NN} }{ \Gamma(l+\frac{b_N N }{2}+1) }
- \sqrt{\pi} \frac{ \mu^{b_NN-1} }{ \Gamma(\frac{b_N N }{2}) }  \sum_{k=0}^{N-1} \frac{ \nu^{2k+1+b_NN} }{  \Gamma(k+\tfrac{3}{2}+\tfrac{b_N N}{2}) }. 
\end{split} \label{lol2}
\end{equation}
Substituting \eqref{sum to Q} in the above expression, we obtain
\begin{equation} 
\begin{split} 
\sqrt{\frac{2}{a_NN}} \pa_\zeta \wh{\bfkappa}_N(\zeta,\eta)&= 2\mu \, \wh{\bfkappa}_N(\zeta,\eta) + 2\, e^{2\mu \nu} \Big( \mathrm{Q}(2N+b_NN,2\mu \nu)-\mathrm{Q}(b_NN, 2\mu \nu) \Big)
\\
&\quad  -2\,  \sqrt{\pi}\, \frac{ \mu^{2N+b_NN} }{  \Gamma(N+\tfrac{1}{2}+\tfrac{b_N N}{2}) }  \,e^{\nu^2} \Big( \mathrm{Q}(N+\tfrac{b_NN}{2}, \nu^2)-\mathrm{Q}(\tfrac{b_NN}{2},\nu^2) \Big)
\\
&\quad -2\, \sqrt{\pi} \, \frac{ \mu^{b_NN-1} }{ \Gamma(\frac{b_N N }{2}) }  \,e^{\nu^2} \Big( \mathrm{Q}(N+\tfrac{b_NN+1}{2}, \nu^2)-\mathrm{Q}(\tfrac{b_NN+1}{2},\nu^2) \Big).
\end{split}
\end{equation}
Using then \eqref{bfkappa wt}, we get the desired identity \eqref{CDI bfkappa}. 
\end{proof}

\section{Proofs of Theorem~\ref{thm:main} and Proposition \ref{Cor_Chiral lim}}\label{Section_Scaling limits}
Parts (a), (b) and (c) of Theorem~\ref{thm:main} are proven in Subsections~\ref{Subsec_cplex limit}, \ref{Subsec_real limit} and \ref{Subsec_scaling interpolation}, respectively, and Proposition~\ref{Cor_Chiral lim} is proved in Subsection~\ref{Subsec_chiral limit}. 

\medskip Recall that $p=e^{i\theta}$ is the zooming point around which we rescale the correlation functions, see \eqref{RNk bfRNk}, and that $z_j$ and $\zeta_j$ are related as \eqref{rescaling}. Recall also that Theorem \ref{thm:main} (a) deals with $\theta \in [0,2\pi)\setminus\{0,\pi\}$ and that Theorem \ref{thm:main} (b) deals with $\theta \approx 0,\pi$. We start this section with general lemma valid for all $\theta$. 
\begin{lemma}\label{lemma:exact identity Rnk}
Let $\theta \in [0,2\pi)$. The following identity holds
\begin{align} \label{RNk bfkappa wt new}
& R_{N,k}(z_1,\dots, z_k) 
= \Pf \Big[ \bigg(\begin{matrix} 
\wt{\bfkappa}_N(\zeta_j,\zeta_l) e^{-\big(|z_{j}|^{2}+|z_{l}|^{2}-2e^{2i\theta}z_{j}z_{l}+(1-e^{2i\theta})a_{N}N+\sqrt{2a_{N}N}(z_{j}+z_{l})(1-e^{2i\theta})\big)} 
\smallskip 
\\
\wt{\bfkappa}_N(\bar{\zeta}_j,\zeta_l)  e^{-\big(|z_{j}|^{2}+|z_{l}|^{2}-2\overline{z}_{j}z_{l}\big)} 
\end{matrix} \cdots  \\
& \cdots \begin{matrix}
\wt{\bfkappa}_N(\zeta_j,\bar{\zeta}_l)   e^{-\big(|z_{j}|^{2}+|z_{l}|^{2}-2z_{j}\overline{z}_{l}\big)} \\
\wt{\bfkappa}_N(\bar{\zeta}_j,\bar{\zeta}_l)   e^{-\big(|z_{j}|^{2}+|z_{l}|^{2}-2e^{-2i\theta}\overline{z}_{j}\overline{z}_{l}+(1-e^{-2i\theta})a_{N}N+\sqrt{2a_{N}N}(\overline{z}_{j}+\overline{z}_{l})(1-e^{-2i\theta})\big)} 
\end{matrix} \bigg) \Big]_{ j,l=1 }^k \prod_{j=1}^{k} \frac{\bar{\zeta}_j-\zeta_j}{\gamma_N}.
\end{align}
\end{lemma}
\begin{proof}
Using \eqref{Rk bfkappa wh} and \eqref{bfkappa wt}, we first rewrite $R_{N,k}$ as
\begin{equation} \label{RNk bfkappa wt}
R_{N,k}(z_1,\dots, z_k) 
= \Pf \Big[ e^{ -\frac{a_NN}{2}(|\zeta_j|^2+|\zeta_l|^2) }   \begin{pmatrix} 
\wt{\bfkappa}_N(\zeta_j,\zeta_l) e^{a_N N \zeta_j \zeta_l} & \wt{\bfkappa}_N(\zeta_j,\bar{\zeta}_l)   e^{a_N N \zeta_j  \bar{\zeta}_l}
\smallskip 
\\
\wt{\bfkappa}_N(\bar{\zeta}_j,\zeta_l)  e^{ a_NN \bar{\zeta}_j \zeta_l } & \wt{\bfkappa}_N(\bar{\zeta}_j,\bar{\zeta}_l)   e^{ a_NN \bar{\zeta}_j \bar{\zeta}_l }
\end{pmatrix}   \Big]_{ j,l=1 }^k \prod_{j=1}^{k} \frac{\bar{\zeta}_j-\zeta_j}{\gamma_N}.
\end{equation}
In the spirit of \eqref{rescaling}, given $z,w \in \mathbb{C}$, let us define $\zeta$ and $\eta$ by
\begin{equation} \label{zeta eta}
\zeta=e^{i\theta} \Big( 1+ \sqrt{ \frac{2}{a_N N} }\,z \Big), \qquad \eta=e^{i\theta} \Big( 1+ \sqrt{ \frac{2}{a_N N} }\,w \Big).
\end{equation}
A direct computation shows that
\begin{align}
\frac{a_NN}{2}(|\zeta|^2+|\eta|^2-2\zeta \eta) &= |z|^2+ |w|^2 -2e^{2i\theta} zw+(1-e^{2i\theta})a_NN \nonumber \\
&\quad +\sqrt{2a_NN}\Big( (z+w) (1-e^{2i\theta})-i \im(z+w) \Big), \label{prefactor 11} \\
\frac{a_NN}{2}(|\zeta|^2+|\eta|^2-2\zeta \bar{\eta} )&=|z|^2+|w|^2-2z\bar{w}-\sqrt{2a_NN} i\im(z+\bar{w}). \label{prefactor 12}
\end{align}
Furthermore, when using \eqref{prefactor 11} and \eqref{prefactor 12} in \eqref{RNk bfkappa wt}, the terms $i\im(z+w)$ and $i\im(z+\bar{w})$ cancel out when computing the Pfaffian in \eqref{RNk bfkappa wt}. The claim follows.
\end{proof}

\subsection{Proof of Theorem~\ref{thm:main} (a)} \label{Subsec_cplex limit}
As in the statement of Theorem~\ref{thm:main} (a), here $\theta \in [0,2\pi) \setminus \{0,\pi \}$. Let us write
\begin{align}
 \label{KN kappa}
K_N(z,w)&:= \frac{\sqrt{2a_NN}\sin \theta}{i} e^{2z\bar{w}} \wt{\bfkappa}_N(\zeta,\bar{\eta}),
\\
\label{eN kappa}
e_N(z,w)&:=  \frac{\sqrt{2a_NN}\sin \theta}{i} e^{  (e^{2i\theta}-1) (a_NN+\sqrt{2a_NN}(z+w) )+2e^{2i\theta} zw } \wt{\bfkappa}_N(\zeta,\eta),
\end{align}
with $\zeta$ and $\eta$ as in \eqref{zeta eta}. 

\begin{proof}[Proof of Theorem~\ref{thm:main} (a)]
Combining Lemma \ref{lemma:exact identity Rnk} with \eqref{prefactor 11}, \eqref{prefactor 12}, \eqref{KN kappa}, \eqref{eN kappa}, and 
	\begin{equation} \label{zeta im}
		\frac{\bar{\zeta}_j-\zeta_j}{\gamma_N} = \frac{ 2\sin \theta }{i \gamma_N}+ e^{-i\theta} \bar{z}_j-e^{i\theta} z_j = \frac{ 2\sin \theta }{i \gamma_N}+o(1), \quad \mbox{as } N \to \infty,
	\end{equation}
	we obtain after a computation that
	\begin{equation} \label{RNk cplx K e}
		R_{N,k}(z_1,\dots,z_k) = \Pf \Big[ e^{-|z_j|^2-|z_l|^2}  \begin{pmatrix}
			e_N(z_j,z_l) & K_N(z_j,z_l)
			\smallskip 
			\\
			-K_N(z_l,z_j) & -\overline{e_N(z_j,z_l)}
		\end{pmatrix} \Big]_{j,l=1}^k \,+o(1),
	\end{equation}
	where the $o(1)$-term is uniform on compact subsets of $\C$. It turns out that
\begin{equation} \label{KN eN asymp in proof}
K_N(z,w) =K^\C(z,w)+o(1), \qquad e_N(z,w)=o(1), \qquad \mbox{as } N \to \infty,
\end{equation}
uniformly for $z,w$ in compact subsets of $\mathbb{C}$, where $K^\C$ is defined in \eqref{K AH}. We postpone the proof of \eqref{KN eN asymp in proof} to Lemma \ref{Lem_eK asym} below. Using \eqref{RNk cplx K e} and \eqref{KN eN asymp in proof}, we obtain 
\begin{equation}
\begin{split}
R_{N,k}(z_1,\dots,z_k) = \Pf \Big[ e^{-|z_j|^2-|z_l|^2}  \begin{pmatrix}
0 & K^\C(z_j,z_l) \smallskip \\
-K^\C(z_l,z_j) & 0
\end{pmatrix} \Big]_{j,l=1}^k \,+o(1), \qquad \mbox{as } N \to \infty
\end{split}
\end{equation}
uniformly for $z_{1},\ldots,z_{k}$ in compact subsets of $\mathbb{C}$. Using row and column operations, we get
\begin{equation}
R_{N,k}(z_1,\dots,z_k) =  (-1)^{\frac{k(k-1)}{2}} \Pf 
\begin{pmatrix}
0 & M \\
-M^T & 0
\end{pmatrix} \,+o(1), \qquad M:=\Big( e^{-|z_j|^2-|z_l|^2} K^\C(z_j,z_l) \Big)_{j,l=1}^k.
\end{equation}
The claim now follows from the algebraic identity 
\begin{equation} \label{Pf block det}
(-1)^{ \frac{k(k-1)}{2}}	\Pf\begin{pmatrix}
0 & M \\
-M^T & 0
\end{pmatrix}= \det(M).
\end{equation}
\end{proof}

We now prove \eqref{KN eN asymp in proof}.
 
\begin{lem} \label{Lem_eK asym}
Let $\theta \in [0,2\pi)\setminus \{0,\pi\}$. As $N \to \infty$, 
\begin{equation} \label{KN eN asymp}
K_N(z,w) =K^\C(z,w)+o(1), \qquad e_N(z,w)=o(1),
\end{equation}
uniformly for $z,w$ in compact subsets of $\C$, where $K^\C$ is defined in \eqref{K AH}.
\end{lem}
\begin{proof}
Recall that $K_N(z,w)$ is defined in \eqref{KN kappa} in terms of $ \wt{\bfkappa}_N(\zeta,\bar{\eta})$, where $\zeta$ and $\eta$ are given by \eqref{zeta eta}.
Recall also from \eqref{bfkappa wt} that 
\begin{equation} \label{bfkappa wt G}
\wt{\bfkappa}_N(\zeta,\overline{\eta}):=e^{-a_NN \zeta \overline{\eta}}\wh{\bfkappa}_N(\zeta,\overline{\eta}) = e^{-a_NN \zeta \overline{\eta}} \Big( \wh{\boldsymbol{G}}_N(\zeta,\overline{\eta})-\wh{\boldsymbol{G}}_N(\overline{\eta},\zeta) \Big),
\end{equation} 
where $\wh{\bfkappa}_N$ and $\wh{\boldsymbol{G}}_N$ are defined in the statement of Lemma~\ref{Lem_bfkappa exp}. The first step of the proof consists in obtaining the large $N$ asymptotics of the summand in the second expression of \eqref{G bold hat} that is valid uniformly for $k\in \{0,\ldots,N-1\}$. For this, we use \eqref{Q aN bN}, \eqref{zeta eta} and Stirling's formula, and obtain after a direct computation that 
\begin{align}
\sqrt{\pi}\frac{ ( \sqrt{\frac{a_NN}{2}}\, \zeta )^{2k+1+b_NN} }{  \Gamma(k+\tfrac{3}{2}+\tfrac{b_N N}{2}) } &= \frac{\rho}{ N } e^{i\theta(2k+1+b_NN)}  \exp\Big( \frac{N^2}{2\rho^2}+ \frac{\sqrt{2}z}{\rho} N-z^2-\sqrt{2}\rho z-\frac{\rho^2}{4} \Big) \nonumber \\
&\quad \times  \exp\Big( \frac{(2\sqrt{2}\rho z+\rho^2) k}{N}-\rho^2 \Big(\frac{k}{N}\Big)^2 \Big)  \Big(1+O(N^{-1})\Big), \quad \mbox{as } N \to \infty \label{lol4}
\end{align}
uniformly for $k\in \{0,\dots,N-1\}$. 
Applying \eqref{Q asymp outer}, we have 
\begin{align}
& e^{\frac{a_NN}{2} \bar{\eta}^2 } \mathrm{Q}\Big(k+\tfrac{b_NN}{2}+1, \tfrac{a_NN}{2} \overline{\eta}^2\Big)  = \frac{1}{ \sqrt{\pi} } \frac{\rho}{N} \frac{ e^{-i\theta  (2k+2+b_NN) } }{ e^{-2i\theta}-1 } \exp\Big( \frac{ (2\sqrt{2} \rho \bar{w}+\rho^2 )k }{N}-\rho^2\Big(\frac{k}{N}\Big)^2  \Big) \nonumber \\
&\quad \times   \exp\Big( \frac{N^2}{ 2\rho^2 }+\frac{\sqrt{2} \bar{w} }{\rho} N-\bar{w}^2-\sqrt{2}\rho\bar{w}-\frac{\rho^2}{4}    \Big)  \Big(1+O(N^{-1})\Big)  \label{ameur cron}
\end{align}
as $N \to \infty$ uniformly for $k\in \{-1,0,\dots,N-1\}$. 
Combining the above equations with 
 \begin{equation}
	a_NN \zeta \bar{\eta} 
	=\frac{N^2}{\rho^2}+ \frac{\sqrt{2}}{ \rho } (z+\bar{w}) N +2 z \bar{w},  \qquad \frac{ e^{-i\theta}   }{ e^{-2i\theta}-1 } = \frac{i}{2\sin \theta}, 
\end{equation}
we obtain 
\begin{align*}
& s_{k,1} := \frac{\sqrt{2a_NN}\sin \theta}{i} e^{2z\bar{w}}  e^{-a_NN \zeta \bar{\eta} } \sqrt{\pi} \frac{ ( \sqrt{\frac{a_NN}{2}}\, \zeta )^{2k+1+b_NN} }{  \Gamma(k+\tfrac{3}{2}+\tfrac{b_N N}{2}) }  \exp\Big( \frac{a_NN}{2}  \bar{\eta}^2 \Big)  \mathrm{Q}\Big(k+\tfrac{b_NN}{2}+1, \tfrac{a_NN}{2}  \bar{\eta}^2\Big) 
\\
&= \frac{1}{\sqrt{2\pi}} \frac{\rho}{N}   \exp\Big( -z^2-\bar{w}^2-\sqrt{2}\rho(z+\bar{w})-\frac{\rho^2}{2}+\frac{ (2\sqrt{2} (z+\bar{w}) \rho+2\rho^2 )k }{N}-2\rho^2\Big(\frac{k}{N}\Big)^2  \Big)   \Big(1+O(N^{-1})\Big). \\
& s_{k,2}:= \frac{\sqrt{2a_NN}\sin \theta}{i} e^{2z\bar{w}}  e^{-a_NN \zeta \bar{\eta} } \sqrt{\pi} \frac{ ( \sqrt{\frac{a_NN}{2}}\, \zeta )^{2k+1+b_NN} }{  \Gamma(k+\tfrac{3}{2}+\tfrac{b_N N}{2}) }  \exp\Big( \frac{a_NN}{2}  \bar{\eta}^2 \Big)  \mathrm{Q}\Big(\tfrac{b_NN}{2}, \tfrac{a_NN}{2}  \bar{\eta}^2\Big) \\
&= \frac{1}{\sqrt{2 \pi}} \frac{\rho}{N}  e^{ i \theta(2k+2) }  \exp\Big( -z^2-\bar{w}^2-\sqrt{2}\rho(z+\bar{w})-\frac{\rho^2}{2}  \Big)   \Big(1+O(N^{-1})\Big),
\end{align*}
as $N \to \infty$ uniformly for $k \in \{0,\ldots,N-1\}$. Since $\sum_{k=0}^{N-1}   e^{ i \theta(2k+2) }  = O(1)$ as $N \to \infty$, we have $\sum_{k=0}^{N-1}s_{k,2} = O(N^{-1})$ as $N \to \infty$. Hence, using a first order Riemann sum approximation, we obtain
\begin{align} 
& \frac{\sqrt{2a_NN}\sin \theta}{i} e^{2z\bar{w}} e^{-a_NN \zeta \bar{\eta} }  \, \wh{\boldsymbol{G}}_N(\zeta,\bar{\eta}) = \sum_{k=0}^{N-1} (s_{k,1}-s_{k,2}) = O(N^{-1}) + \sum_{k=0}^{N-1} s_{k,1} \nonumber \\
& = \frac{1}{\sqrt{2\pi}}\rho \exp\Big( -z^2-\bar{w}^2-\sqrt{2}\rho(z+\bar{w})-\frac{\rho^2}{2}\Big) \int_0^1 \exp\Big( (2\sqrt{2} (z+\bar{w}) \rho+2\rho^2 )x-2\rho^2 x^2  \Big) \,dx + O(N^{-1}). \nonumber
\end{align}
The above integral can be evaluated explicitly,
\begin{equation}
\int_0^1 \exp\Big( (2\sqrt{2} (z+\bar{w}) \rho+2\rho^2 )x-2\rho^2 x^2  \Big) \,dx= \sqrt{ \frac{\pi}{2} } \frac{1}{\rho} \exp\Big(  z^2+\bar{w}^2+\sqrt{2}\rho(z+\bar{w})+\frac{\rho^2}{2} \Big) K^{\mathbb{C}}(z,w),
\end{equation} 
where $K^{\mathbb{C}}$ is defined in \eqref{K AH}, and therefore
\begin{align} 
\frac{\sqrt{2a_NN}\sin \theta}{i} e^{2z\bar{w}} e^{-a_NN \zeta \bar{\eta} }  \, \wh{\boldsymbol{G}}_N(\zeta,\bar{\eta}) & = \frac{1}{2} K^{\mathbb{C}}(z,w)+O(N^{-1}). \label{G bold hat sum zeta bar eta}
\end{align}
By interchanging $\zeta$ and $\bar{\eta}$ in the above computations and using $e^{i\theta}/(e^{2i\theta}-1)=-i/(2\sin \theta)$, we also obtain 
\begin{equation}  \label{G bold hat sum bar eta zeta}
\frac{\sqrt{2a_NN}\sin \theta}{i} e^{2z\bar{w}}  e^{-a_NN \zeta \bar{\eta} }  \, \wh{\boldsymbol{G}}_N(\bar{\eta},\zeta)= -\frac{1}{2} K^{\mathbb{C}}(z,w)+O(N^{-1}).
\end{equation}
Then by combining \eqref{KN kappa}, \eqref{bfkappa wt G}, \eqref{G bold hat sum zeta bar eta} and \eqref{G bold hat sum bar eta zeta}, we obtain the first asymptotics of \eqref{KN eN asymp}. To prove the second part of \eqref{KN eN asymp}, we use the formula \eqref{eN kappa} for $e_{N}$ together with \eqref{bfkappa wt G} (with $\overline{\eta}$ replaced by $\eta$) to write
\begin{align}\label{eN in proof}
 \hspace{0.7cm}  e_N(z,w):=  \frac{\sqrt{2a_NN}\sin \theta}{i} e^{  (e^{2i\theta}-1) (a_NN+\sqrt{2a_NN}(z+w) )+2e^{2i\theta} zw } e^{-a_NN \zeta \eta} \Big( \wh{\boldsymbol{G}}_N(\zeta,\eta)-\wh{\boldsymbol{G}}_N(\eta,\zeta) \Big).
\end{align}
The asymptotics of $\wh{\boldsymbol{G}}_N$ can be obtained using \eqref{lol4} and \eqref{ameur cron} (with $\overline{\eta}$, $\overline{w}$ and $-\theta$ replaced by $\eta$, $w$ and $\theta$, respectively). The above exponential $e^{  (e^{2i\theta}-1) (a_NN+\sqrt{2a_NN}(z+w) )+2e^{2i\theta} zw }$ get perfectly cancelled in the asymptotics; indeed, using $a_NN \zeta \eta = e^{2i\theta} \big( \frac{N^2}{\rho^2}+ \frac{\sqrt{2}}{ \rho } (z+w) N +2 z w \big)$, we obtain
\begin{align*}
& \frac{\sqrt{2a_NN}\sin \theta}{i} e^{  (e^{2i\theta}-1) (a_NN+\sqrt{2a_NN}(z+w) )+2e^{2i\theta} zw } e^{-a_NN \zeta \eta} \wh{\boldsymbol{G}}_N(\zeta,\eta) \\
& = \frac{1}{N}\frac{\sqrt{2} \, \rho \sin \theta}{ \sqrt{\pi} \, i} \sum_{k=0}^{N-1}  e^{i\theta(2k+1+b_NN)} \exp\Big( -z^2-w^2-\sqrt{2}\rho (z+w)-\frac{\rho^2}{2} \Big) \nonumber \\
&\quad \times \bigg[ \exp\Big( \frac{(2\sqrt{2}\rho (z+w)+2\rho^2) k}{N}-2 \Big(\frac{\rho k}{N}\Big)^2 \Big)  \frac{ e^{i\theta  (2k+2+b_NN) } }{ e^{2i\theta}-1 } -  \frac{ e^{i\theta  b_NN } }{ e^{2i\theta}-1 }  \bigg]  \Big(1+O(N^{-1})\Big),
\end{align*}
as $N \to +\infty$ uniformly for $k \in \{0,\ldots,N-1\}$. By \eqref{eN in proof}, we thus have
\begin{align*}
e_{N}(z,w) = \frac{1}{N} O \bigg( \sum_{k=0}^{N-1} N^{-1} \bigg) = O(N^{-1}), \qquad \mbox{as } N \to + \infty.
\end{align*}
\end{proof}

\subsection{Proof of Theorem~\ref{thm:main} (b)}\label{Subsec_real limit}
Recall that $t \in \mathbb{R}$ and $\theta_{N} := \gamma_{N}t = \frac{\sqrt{2} \, \rho}{N}t$.  The two cases $p=e^{i\theta_N}$ and $p=-e^{i\theta_N}$ are similar, so to avoid repetition we will only consider the case $p:=e^{i\theta_{N}}$. 
Recall that $a_N$ and $b_N$ are defined by \eqref{Q aN bN}. In this subsection, given $z, w \in \mathbb{C}$, we define $\zeta$ and $\eta$ as in \eqref{zeta eta} with $\theta= \theta_{N}$, namely
\begin{align}\label{zeta eta theta N}
\zeta=e^{i\theta_{N}} \Big( 1+ \sqrt{ \frac{2}{a_N N} }\,z \Big), \qquad \eta=e^{i\theta_{N}} \Big( 1+ \sqrt{ \frac{2}{a_N N} }\,w \Big),
\end{align}
so that $\mu$ and $\nu$ in \eqref{mu nu} become 
\begin{equation} \label{mu nu N other}
\mu:=\sqrt{ \frac{a_NN}{2} } \zeta =e^{i\theta_{N}}\Big(z+ \sqrt{ \frac{a_N N}{2} }\Big), \qquad \nu:=\sqrt{ \frac{a_NN}{2} } \eta =e^{i\theta_{N}}\Big(w+ \sqrt{ \frac{a_N N}{2} }\Big).
\end{equation}
In particular, $\zeta,\eta,\mu,\nu$ always depend on $N$ in this subsection, although this is not indicated in the notation. 
The following lemma is a rather direct consequence of Proposition~\ref{Prop_CDI} and Lemma \ref{lemma:exact identity Rnk}. 

\begin{lem} \label{Lem_CDI real}
Let $t \in \mathbb{R}$ and $\rho>0$ be fixed, and $p:=e^{i\theta_{N}}$. As $N \to \infty$
\begin{align}
R_{N,k}(z_1,\dots, z_k) & = \Pf \Big[ e^{ -|z_j+it|^2-|z_l+it|^2 }   \begin{pmatrix} 
\wt{\bfkappa}_N(\zeta_j,\zeta_l) e^{2(z_j+it)(z_l+it)} & \wt{\bfkappa}_N(\zeta_j,\bar{\zeta}_l) e^{2(z_j+it)(\bar{z}_l-it)} \smallskip \\
\wt{\bfkappa}_N(\bar{\zeta}_j,\zeta_l) e^{2(\bar{z}_j-it)(z_l+it) } & \wt{\bfkappa}_N(\bar{\zeta}_j,\bar{\zeta}_l) e^{ 2 (\bar{z}_j-it)(\bar{z}_l-it)  }
\end{pmatrix}   \Big]_{ j,l=1 }^k  \nonumber
\\
&\quad \times \prod_{j=1}^{k} (\bar{z}_j-z_j-2it) + o(1), 	\label{lol5}
\end{align}
uniformly for $z_{1},\ldots,z_{k}$ in compact subsets of $\mathbb{C}$, where $\zeta_{j}:= p+ p \gamma_N z_j$ for $j=1,\ldots,k$. Furthermore, for any $z,w \in \mathbb{C}$, we have
\begin{align}
e^{-i\theta_{N}} \frac{d}{dz} \wt{\bfkappa}_N(\zeta,\eta)&= 2e^{i\theta_{N}}(z-w) \wt{\bfkappa}_N(\zeta,\eta) + 2 \Big( \mathrm{Q}(2N+b_NN,2\mu \nu)-\mathrm{Q}(b_NN, 2\mu \nu) \Big) \nonumber 
\\
& - 2\sqrt{\pi} \,e^{e^{2i\theta_{N}}(z-w)^2}\frac{ \mu^{2N+b_N N}   }{    \Gamma(N+\tfrac{1}{2}+\tfrac{b_N N}{2}) } \,e^{-\mu^2} \Big( \mathrm{Q}(N+\tfrac{b_NN}{2}, \nu^2)-\mathrm{Q}(\tfrac{b_NN}{2},\nu^2) \Big) \nonumber
\\
& - 2  \sqrt{\pi} \,e^{e^{2i\theta_{N}}(z-w)^2} \frac{  \mu^{b_N N-1} }{ \Gamma(\frac{b_N N }{2}) } \,e^{-\mu^2}  \Big( \mathrm{Q}(N+\tfrac{b_NN+1}{2}, \nu^2)-\mathrm{Q}(\tfrac{b_NN+1}{2},\nu^2) \Big), \label{CDI kappa tilde}
\end{align}
where $\zeta =\zeta(z)$ and $\eta=\eta(w)$ are as in \eqref{zeta eta theta N}.

\end{lem}

\begin{proof}
The expansion \eqref{lol5} directly follows from Lemma \ref{lemma:exact identity Rnk} with $\theta=\theta_{N}$. To obtain the exponentials inside the Pfaffian in \eqref{lol5}, we have used the following large $N$ expansion 
\begin{align*}
&|z_{j}|^{2}+|z_{l}|^{2}-2e^{2i\theta_N}z_{j}z_{l}+(1-e^{2i\theta_N})a_{N}N+\sqrt{2a_{N}N}(z_{j}+z_{l})(1-e^{2i\theta_N})\\
&=  |z_{j}+it|^{2}+|z_{l}+it|^{2}-2 (z_{j}+it)(z_{l}+it) -2it \Re(z_j+z_l)  -\frac{ 2\sqrt{2} t N }{\rho } i +o(1),
\end{align*}
the identity $|z_{j}|^{2}+|z_{l}|^{2}-2z_{j}\overline{z}_{l} =  |z_{j}+it|^{2}+|z_{l}+it|^{2}-2(z_{j}+it)(\overline{z}_{l}-it)- 2it \Re(z_j-z_l)$, and the fact that the terms containing $\Re(z_j+z_l)$, $\Re(z_j-z_l)$ and $\frac{2\sqrt{2}tN}{\rho}i$ cancel out when computing the Pfaffian. 
The differential identity \eqref{CDI kappa tilde} immediately follows from \eqref{zeta eta theta N}, $\partial_{\zeta} = e^{-i\theta_{N}} \sqrt{ \frac{a_N N}{2} }\partial_{z}$ and Proposition~\ref{Prop_CDI}.
\end{proof}

We now derive the large $N$ asymptotics of $\pa_z \wt{\bfkappa}_N(\zeta,\eta)$ using the right-hand side of \eqref{CDI kappa tilde}. 

\begin{lem} \label{Lem_CDI largeN real}
Let $t \in \mathbb{R}$ and $\rho>0$ be fixed. Let $z_{t}:=z+it$ and $w_{t}:=w+it$. As $N \to \infty$, we have
\begin{align}
\pa_z \wt{\bfkappa}_N(\zeta,\eta)&= 2(z_{t}-w_{t}) \wt{\bfkappa}_N(\zeta,\eta) +  \erfc(z_{t}+w_{t}-\tfrac{\rho}{\sqrt{2}})-  \erfc(z_{t}+w_{t}+\tfrac{\rho}{\sqrt{2}})  \nonumber \\
& -  \frac{e^{(z_{t}-w_{t})^2}}{\sqrt{2}} \Big( e^{ -(\sqrt{2}z_{t}-\frac{\rho}{2})^2 }+ e^{ -(\sqrt{2}z_{t}+\frac{\rho}{2})^2 }  \Big)\Big( \erfc(\sqrt{2}w_{t}-\tfrac{\rho}{2})-\erfc(\sqrt{2}w_{t}+\tfrac{\rho}{2}) \Big)+o(1), \label{kappa tilde ODE lim}
\end{align}
uniformly for $z,w$ in compact subsets of $\C$, where $\zeta$ and $\eta$ are as in \eqref{zeta eta theta N}.
\end{lem}

\begin{proof}[Proof of Lemma~\ref{Lem_CDI largeN real}]
By \eqref{Q aN bN} and \eqref{mu nu}, we have 
\begin{align*}
& 2N+b_NN= \frac{N^2}{\rho^2}+N, & & b_NN=\frac{N^2}{\rho^2}-N, \\
& 2\mu \nu= \frac{N^2}{\rho^2}+\frac{\sqrt{2}}{\rho} (z_{t}+w_{t})N+2z_{t}w_{t}+o(1), & & \nu^2= \frac{N^2}{2\rho^2}+\frac{\sqrt{2}}{\rho} w_{t} N+w_{t}^2 +o(1)
\end{align*}
as $N \to \infty$ uniformly for $z,w$ in compact subsets of $\C$. By \cite[Eq.(8.11.10)]{olver2010nist}, 
\begin{equation} \label{Q erfc}
\mathrm{Q}(s+1,s+\sqrt{2s}z)= \frac12 \erfc(z)+ \frac{1}{3} \sqrt{ \frac{2}{\pi } } (1+z^2)e^{-z^2} \frac{1}{\sqrt{s}} +  O(1/s), \qquad s\to +\infty 
\end{equation}
uniformly for $z$ in compact subsets of $\C$. 
It readily follows from \eqref{Q erfc} that
\begin{align}
& \mathrm{Q}(2N+b_NN,2\mu \nu)-\mathrm{Q}(b_NN, 2\mu \nu) = \frac12 \Big( \erfc(z_{t}+w_{t}-\tfrac{\rho}{\sqrt{2}})-  \erfc(z_{t}+w_{t}+\tfrac{\rho}{\sqrt{2}}) \Big)+o(1), \label{inho1} \\
& \mathrm{Q}(N+\tfrac{b_NN}{2}, \nu^2)-\mathrm{Q}(\tfrac{b_NN}{2},\nu^2) = \frac12 \Big( \erfc(\sqrt{2}w_{t}-\tfrac{\rho}{2})-\erfc(\sqrt{2}w_{t}+\tfrac{\rho}{2}) \Big) +o(1), \label{inho2} \\
& \mathrm{Q}(N+\tfrac{b_NN+1}{2}, \nu^2)-\mathrm{Q}(\tfrac{b_NN+1}{2},\nu^2)  = \frac12 \Big( \erfc(\sqrt{2}w_{t}-\tfrac{\rho}{2})-\erfc(\sqrt{2}w_{t}+\tfrac{\rho}{2}) \Big)+o(1), \label{inho3}
\end{align}
as $N \to \infty$ uniformly for $z,w$ in compact subsets of $\C$. Also, by Stirling's formula, we have
\begin{equation} \label{inho2-2 3}
\frac{ \mu^{2N+b_N N}  e^{-\mu^2} }{    \Gamma(N+\tfrac{1}{2}+\tfrac{b_N N}{2}) } = \frac{e^{ -(\sqrt{2}z_{t}-\frac{\rho}{2})^2 }}{\sqrt{2\pi}} +o(1), \qquad \frac{  \mu^{b_N N-1} }{ \Gamma(\frac{b_N N }{2}) } \,e^{-\mu^2}  = \frac{e^{ -(\sqrt{2}z_{t}+\frac{\rho}{2})^2 }}{\sqrt{2\pi}} +o(1),
\end{equation}
as $N \to \infty$ uniformly for $z$ in compact subsets of $\C$. Combining \eqref{inho1}, \eqref{inho2}, \eqref{inho3} and \eqref{inho2-2 3}, we obtain \eqref{kappa tilde ODE lim}.
\end{proof}

We now finish the proof of Theorem~\ref{thm:main} (b).

\begin{proof}[Proof of Theorem~\ref{thm:main} (b)]
The proof can be summarized as follows: we first use Lemma \ref{Lem_CDI largeN real} to obtain large $N$ asymptotics for $\wt{\bfkappa}_N(\zeta,\eta)$. We then substitute these asymptotics in \eqref{lol5} to obtain the leading order large $N$ behavior of $R_{N,k}(z_1,\dots, z_k)$.

\medskip In Lemma~\ref{Lem_CDI largeN real}, we have derived \eqref{kappa tilde ODE lim}, which can be seen as a family of ODE (indexed by $N$) of the form $\partial_{z}\wt{\bfkappa}_N = c_{0}(z) \wt{\bfkappa}_N + c_{1}(z) + \mathcal{E}_{N}(z)$ where $\mathcal{E}_{N}(z) \to 0$ as $N \to + \infty$ uniformly for $z$ in compact subsets of $\mathbb{C}$.
By \cite[Lemma 3.10]{byun2021universal}, the limit 
\begin{align*}
\widetilde{\kappa}(z,w) := \lim_{N\to +\infty} \wt{\bfkappa}_N(\zeta,\eta) =  \lim_{N\to +\infty} \wt{\bfkappa}_N\big(e^{i\theta_{N}}(1+\gamma_{N}z),e^{i\theta_{N}}(1+\gamma_{N}w)\big)
\end{align*}
exists for all $z\in \mathbb{C}$, is analytic and satisfies $\partial_{z}\widetilde{\kappa}=c_{0}(z)\widetilde{\kappa}+c_{1}(z)$ and $\partial_{z}\widetilde{\kappa}|_{z=w}=0$. More precisely, we have
\begin{align}
\pa_z \wt{\kappa}(z,w) & = 2(z_{t}-w_{t}) \wt{\kappa}(z,w) +  \erfc(z_{t}+w_{t}-\tfrac{\rho}{\sqrt{2}})-  \erfc(z_{t}+w_{t}+\tfrac{\rho}{\sqrt{2}}) \nonumber \\
& -  \frac{e^{(z_{t}-w_{t})^2}}{\sqrt{2}} \Big( e^{ -(\sqrt{2}z_{t}-\frac{\rho}{2})^2 }+ e^{ -(\sqrt{2}z_{t}+\frac{\rho}{2})^2 }  \Big)\Big( \erfc(\sqrt{2}w_{t}-\tfrac{\rho}{2})-\erfc(\sqrt{2}w_{t}+\tfrac{\rho}{2}) \Big). \label{ODE wt kappa}
\end{align}
For a given $w$, we view \eqref{ODE wt kappa} as a first order ODE in $z$ with the initial condition $\wt{\kappa}(w,w)=0$. Since $c_{0}(z)$ and $c_{1}(z)$ are analytic, uniqueness of the solution to this ODE follows from standard theory. 

To obtain the solution of \eqref{ODE wt kappa}, we first rewrite it as
\begin{align}
\pa_z \wt{\kappa}(z,w) & = 2(z_{t}-w_{t}) \wt{\kappa}(z,w) + e^{(z-w)^{2}}\big(\partial_{z}F_{1}(z_{t},w_{t}) + \partial_{z}F_{2}(z_{t},w_{t})\big), \label{ODE wt kappa new}
\end{align}
where
\begin{align*}
& F_{1}(z,w):=\frac{ 1 }{\sqrt{2}} \int_{-a}^a  \big( e^{-2(z-u)^2  } \erfc( \sqrt{2}(w-u) )- e^{-2(w-u)^2  } \erfc( \sqrt{2}(z-u) ) \big) du, \\
& F_{2}(z,w) := \frac{\sqrt{\pi}}{4}  \Big[\erfc(\sqrt{2}(z+a)) \erfc(\sqrt{2}(w-a)) - \erfc(\sqrt{2}(z-a)) \erfc(\sqrt{2}(w+a))  \Big].
\end{align*}
Indeed, using integration by parts, we obtain
\begin{align}
\partial_z F_{1}(z,w) & =
e^{-(z-w)^2} \Big( \erfc(z+w-2a)-\erfc(z+w+2a) \Big) \nonumber \\
&-\frac{ 1   }{ \sqrt{2} } \Big( e^{-2(z-a)^2} \erfc(\sqrt{2}(w-a)) - e^{-2(z+a)^2}\erfc( \sqrt{2}(w+a) ) \Big), \label{F deri}  \\
\partial_z F_{2}(z,w) & = \frac{1}{\sqrt{2}} \Big( e^{-2(z-a)^2} \erfc(\sqrt{2}(w+a)) - e^{-2(z+a)^2}\erfc( \sqrt{2}(w-a) ) \Big). \label{F2 deri}
\end{align}
Note that $F_{1}(w,w)=F_{2}(w,w)=0$. It is now readily checked that the unique solution of \eqref{ODE wt kappa new} satisfying $\wt{\kappa}(w,w)=0$ is given by
\begin{equation}
\wt{\kappa}(z,w):=e^{(z-w)^2} \Big[ F_{1}(z_{t},w_{t}) +F_{2}(z_{t},w_{t})  \Big] = e^{-2z_{t}w_{t}} \kappa^{\mathbb{R}}(z_{t},w_{t}),
\end{equation}
where $\kappa^{\mathbb{R}}(z,w)$ is given by \eqref{kappa Wronskian}. 
\end{proof}

\subsection{Proof of Theorem~\ref{thm:main} (c)} \label{Subsec_scaling interpolation}
We first obtain the large $t$ asymptotics of 
\begin{align*}
e^{-|z+it|^2-|w+it|^2}\kappa^{\mathbb{R}}(z+it,\overline{w}-it) \qquad \mbox{ and } \qquad e^{-|z+it|^2-|w+it|^2}\kappa^{\mathbb{R}}(z+it,w+it),
\end{align*}
where $\kappa^{\mathbb{R}}$ is defined in \eqref{kappa Wronskian}. Using the well-known $z\to + \infty$ asymptotics of $\erfc(z)$ (see e.g. \cite[Eq.(7.12.1)]{olver2010nist}), we obtain
\begin{align*}
& \sqrt{\pi} e^{-|z+it|^2-|w+it|^2} e^{(z+it)^2+(\bar{w}-it)^2}  \int_{-a}^a W(f_{\bar{w}-it},f_{z+it})(u) \, du =  \frac{i}{2t}   e^{ -|z|^2-|w|^2 } K^\C(z,w) \, c(z,w) + O(t^{-2}),
\end{align*}
as $t \to \infty$ uniformly for $z$ and $w$ in compact subsets of $\mathbb{C}$, where $c(z,w):=e^{ -(z+\bar{z})it+(w+\bar{w})it }$ satisfies $c(z,w)=1/c(w,z)$ is a therefore an unimportant cocycle. We also have  
\begin{equation}
\sqrt{\pi}	e^{-|z+it|^2-|w+it|^2} e^{(z+it)^2+(\bar{w}-it)^2}   \Big( f_w(a)f_z(-a) -f_z(a)f_w(-a) \Big) =O(t^{-2}),  \quad \mbox{as } t \to \infty.
\end{equation}
Then by \eqref{kappa Wronskian}, we have 
\begin{equation}
e^{-|z+it|^2-|w+it|^2}   \kappa^\R (z+it,\bar{w}-it) = \frac{i}{2t} e^{ -|z|^2-|w|^2 } K^\C(z,w) \, c(z,w) \,+O(t^{-2}), \quad \mbox{as } t \to \infty,
\end{equation}
uniformly for $z$ and $w$ in compact subsets of $\mathbb{C}$. Similarly, we have 
\begin{equation}
e^{-|z+it|^2-|w+it|^2}  \kappa^\R (z+it,w+it) =O(t^{-2}), \quad \mbox{as } t \to \infty,
\end{equation}
uniformly for $z$ and $w$ in compact subsets of $\mathbb{C}$. Let us write 
\begin{equation}
M:=\Big( c(z_j,z_l)\,e^{-|z_j|^2-|z_l|^2} K^\C(z_j,z_l) \Big)_{j,l=1}^k. 
\end{equation}
Combining the above expansions with \eqref{RNk near real} and \eqref{K AH}, and performing elementary row and column operations, we obtain that as $t \to \infty$, 
\begin{equation}
\begin{split}
R_k^{\mathbb{R}}(z_1+it,\dots,z_k+it) & =  (-1)^{\frac{k(k-1)}{2}} \Pf 
\begin{pmatrix}
0 & M\\
-M^T & 0
\end{pmatrix} +o(1), \\
&=  \det (M) +o(1) = \det \Big( e^{-|z_j|^2-|z_l|^2} K^\C(z_j,z_l) \Big)_{j,l=1}^k +o(1), 
\end{split}
\end{equation}
which gives \eqref{Pf to det}. Here, the second identity follows from  \eqref{Pf block det}, whereas the third one follows from the fact that the cocycle $c$ cancels out when forming the determinant. 
The proof is complete.

\subsection{Proof of Proposition~\ref{Cor_Chiral lim}}\label{Subsec_chiral limit}
We start with an auxiliary lemma. 

\begin{lem} \label{Lem_Pf det iden}
For any function $f: \C^2 \to \C$ satisfying $f(x,y)=f(y,x)$, any $x_{1},\ldots,x_{k}\in \mathbb{C}$ and any $y_{1},\ldots,y_{k} \in \mathbb{C}\setminus \{0\}$, we have
\begin{equation} \label{Pf det iden}
\Pf \Big[ f(x_j,x_l) \begin{pmatrix} 
\frac{y_j-y_l}{2y_j y_l}  & \frac{y_j+y_l}{2 y_j y_l } \smallskip \\
-\frac{y_j+y_l}{2 y_j y_l }  & -\frac{y_j-y_l}{2 y_j y_l } \end{pmatrix} \Big]_{j,l=1}^k \, \prod_{j=1}^k y_j = \det \Big[ f(x_j,x_l) \Big]_{j,l=1}^k.
\end{equation}
\end{lem}
\noindent To be concrete, for $k=1,2$ formula \eqref{Pf det iden} reads as follows: letting $f_{jk}:=f(x_j,x_k)$, 
\begin{equation}
	\Pf  \begin{pmatrix}
	0 & f_{11} \frac{1}{y_1}
	\\
	-f_{11} \frac{1}{y_1} & 0
\end{pmatrix} y_1 = f_{11}, \; \Pf  \begin{pmatrix}
	0 & f_{11} \frac{1}{y_1} &  f_{12} \frac{y_1-y_2}{2 y_1y_2 }  &  f_{12} \frac{y_1+y_2}{2y_1y_2}
	\smallskip 
	\\
	-f_{11} \frac{1}{y_1} & 0 & - f_{12} \frac{y_1+y_2}{2y_1y_2} & - f_{12} \frac{y_1-y_2}{2y_1y_2}
		\smallskip 
	\\
	-f_{12} \frac{y_1-y_2}{2y_1y_2} & f_{12} \frac{y_1+y_2}{2y_1y_2} & 0 & f_{22} \frac{1}{y_2}
		\smallskip 
	\\
	-f_{12} \frac{y_1+y_2}{2y_1y_2} & f_{12} \frac{y_1-y_2}{2y_1y_2} & -f_{22}\frac{1}{y_2} & 0
\end{pmatrix}  y_1y_2=  f_{11} f_{22}-f_{12}^2.
\end{equation}

\begin{proof}[Proof of Lemma~\ref{Lem_Pf det iden}]
	
	Let 
	\begin{equation}
	F:= \Big[ f_{jl}  \begin{pmatrix} 
		\frac{y_j-y_l}{2y_j y_l}  & \frac{y_j+y_l}{2 y_j y_l } 
		\smallskip 
		\\
		-\frac{y_j+y_l}{2 y_j y_l }  & -\frac{y_j-y_l}{2 y_j y_l } 
	\end{pmatrix}    \Big]_{j,l=1}^k=  \Big[ \frac{f_{jl}}{2}  \begin{pmatrix} 
	\frac{1}{y_l}-\frac{1}{y_j}  & \frac{1}{y_l}+\frac{1}{y_j} 
	\smallskip 
	\\
	-\frac{1}{y_l}-\frac{1}{y_j}   & 	\frac{1}{y_j}-\frac{1}{y_l} 
\end{pmatrix}    \Big]_{j,l=1}^k . 
	\end{equation}
	We first show that 
\begin{equation} \label{Pf y indep}
	\Pf \,(F) \prod_{j=1}^k y_j= \Pf  \Big[ f_{jl}  \begin{pmatrix} 
		0  & 1
		\smallskip 
		\\
		-1  & 0
	\end{pmatrix}    \Big]_{j,l=1}^k .
\end{equation}
For this, recall that for any $k\times k$ matrix $A=(a_{j,l})_{j,l=1}^k$ and any $2k \times 2k$  skew-symmetric matrix $B$,  
\begin{equation} \label{det def}
\det (A)= \sum_{ \sigma \in S_k } \textup{sgn} (\sigma) \prod_{j=1}^k a_{ j,\sigma(j) }, \qquad \Pf\,(B)^2=\det (B),
\end{equation}
where $S_k$ is the symmetric group of all permutations of size $k$. Using row and column operations, we observe that
\begin{equation}
\begin{split}
\det (F) & = \det  \Big[ \frac{f_{jl}}{2}  \begin{pmatrix} 
y_l^{-1}-y_j^{-1}  & y_l^{-1}+y_j^{-1} \smallskip \\
-y_l^{-1}-y_j^{-1} & y_j^{-1}-y_l^{-1}
\end{pmatrix} \Big]_{j,l=1}^k =  \det  \Big[ \frac{f_{jl}}{2}  \begin{pmatrix} 
-2y_j^{-1}  & 2y_j^{-1} \smallskip \\
-y_l^{-1}-y_j^{-1} & y_j^{-1}-y_l^{-1}
\end{pmatrix} \Big]_{j,l=1}^k
\\
&= \det  \Big[ \frac{f_{jl}}{2}  \begin{pmatrix} 
0  & 2y_j^{-1} \smallskip \\
-2y_l^{-1}   & y_j^{-1}-y_l^{-1}
\end{pmatrix} \Big]_{j,l=1}^k =  \det  \Big[ \frac{f_{jl}}{2} \begin{pmatrix} 
0  & 2y_j^{-1} \smallskip \\
-2y_l^{-1}  & -y_l^{-1}
\end{pmatrix}  \Big]_{j,l=1}^k  \\
&=   \det  \Big[ f_{jl}  \begin{pmatrix} 
0  & y_j^{-1} \smallskip \\
-y_l^{-1} & 0
\end{pmatrix} \Big]_{j,l=1}^k = \det  \Big[ f_{jl}  \begin{pmatrix} 
0  & 1 \smallskip \\
-1 & 0
\end{pmatrix} \Big]_{j,l=1}^k  \prod_{j=1}^k y_j^{-2}. 
    \end{split}
    \end{equation}
Combining the above with the second identity in \eqref{det def} yields \eqref{Pf y indep}. Finally, using again row and column operations, we obtain 
\begin{equation}\label{lol6}
\begin{split}
\Pf  \Big[ f_{jl} \begin{pmatrix} 
0  & 1 \smallskip  \\
-1  & 0 \end{pmatrix} \Big]_{j,l=1}^k  &=  (-1)^{ \frac{k(k-1)}{2} } \Pf  \Big[   \begin{pmatrix} 
0  & M \smallskip  \\
-M^T  & 0
\end{pmatrix} \Big]_{j,l=1}^k , \qquad M:= (  f_{jl} )_{j,l=1}^k.
\end{split}
\end{equation}
The desired identity \eqref{Pf det iden} follows from \eqref{Pf y indep}, \eqref{lol6} and \eqref{Pf block det}. 
\end{proof}

We are now ready to prove Proposition~\ref{Cor_Chiral lim}.

\begin{proof}[Proof of Proposition~\ref{Cor_Chiral lim}]

Throughout the proof, let $x_j:= \re z_j$ and $y_j=\im z_j$ ($j=1,\dots,k$). Using the definition of $R_{k}^{\mathbb{R}}$, the left-hand side of \eqref{convergence in distribution} can be rewritten as
\begin{equation}\label{lol9}
\wt{R}_k(z_1,\dots,z_k):=\frac{1}{a^{2k}} R_k^{\mathbb{R}}(\frac{z_1}{a},\dots, \frac{z_k}{a}) = \frac{1}{a^{3k}}  \Pf \Big[ e^{-\frac{|z_j|^2+|z_l|^2}{a^2}} 
\begin{pmatrix} 
\kappa^{\mathbb{R}}(\frac{z_j}{a},\frac{z_l}{a}) & \kappa^{\mathbb{R}}(\frac{z_j}{a},\frac{\bar{z}_l}{a})  \smallskip  \\
\kappa^{\mathbb{R}}(\frac{\bar{z}_j}{a},\frac{z_l}{a}) & \kappa^{\mathbb{R}}(\frac{\bar{z}_j}{a},\frac{\bar{z}_l}{a}) 
\end{pmatrix} 
\Big]_{j,l=1}^k \prod_{j=1}^k (-2i \, y_j).
\end{equation} 
Let us first compute the leading order behavior of $\frac{1}{a^3} e^{-\frac{|z|^2+|w|^2}{a^2}} \kappa^{\mathbb{R}}(\frac{z}{a},\frac{w}{a})$ as $a \to 0$. Using the well-known $z\to + \infty$ asymptotics of $\erfc(z)$ (see e.g. \cite[Eq.(7.12.1)]{olver2010nist}), we find  
\begin{equation}
f_{z/a}(u)=\frac12 \erfc( \sqrt{2}( \tfrac{z}{a}-u ) ) =  \frac{1}{2\sqrt{2\pi}} \frac{a}{z} \, e^{ -2 ( \tfrac{z}{a}-u )^2 } \,(1+o(1)), \qquad \mbox{as } a \to 0,
\end{equation}
uniformly for $u$ in compact subsets of $\mathbb{C}$. Thus we have 
\begin{align}
& \frac{1}{a^{3}}\sqrt{\pi}e^{-\frac{|z|^2+|w|^2-z^2-w^2}{a^2}} \Big( f_{w/a}(a)f_{z/a}(-a)-f_{z/a}(a)f_{w/a}(-a) \Big) \nonumber
\\
& =  -\sqrt{\pi}e^{-\frac{|z|^2+|w|^2+z^2+w^2}{a^2}} \frac{1}{8\pi} \frac{1}{azw}\Big(  e^{4(z-w)}-e^{-4(z-w)}  \Big) \,(1+o(1)), \qquad \mbox{as } a \to 0. \label{lol7}
\end{align}
On the other hand, 
\begin{equation}
\begin{split}
W(f_{w/a},f_{z/a})(u)&= -\sqrt{ \frac{2}{\pi} } \Big(  e^{-2( \frac{w}{a}-u )^2 } f_{z/a}(u)- e^{-2( \frac{z}{a}-u )^2 } f_{w/a}(u)  \Big)
\\
&= -\frac{a}{2\pi} \Big( \frac{1}{z}-\frac{1}{w}    \Big) \, e^{-2( \frac{z}{a}-u )^2-2( \frac{w}{a}-u )^2 } \,(1+o(1)), \quad \mbox{as } a \to 0
\end{split}
\end{equation}
uniformly for $u$ in compact subsets of $\mathbb{C}$. Thus we have 
\begin{equation}
\begin{split}
\frac{1}{a^{3}}\sqrt{\pi}e^{ -\frac{|z|^2+|w|^2-z^2-w^2}{a^2} }W(f_{w/a},f_{z/a})(au) 
&= -\frac{1}{2a^{2}\sqrt{\pi}} \Big( \frac{1}{z}-\frac{1}{w}    \Big) \,e^{ -\frac{|z|^2+|w|^2+z^2+w^2}{a^2} }\, e^{4u(z+w)} \,(1+o(1)).
\end{split}
\end{equation}
This gives  
\begin{equation}
\begin{split}
& \frac{1}{a^{3}}\sqrt{\pi} e^{ -\frac{|z|^2+|w|^2-z^2-w^2}{a^2} } \int_{-a}^a W(f_{w/a},f_{z/a})(u)\,du = \frac{\sqrt{\pi}}{a^{2}} e^{ -\frac{|z|^2+|w|^2-z^2-w^2}{a^2} }  \int_{-1}^1 W(f_{w/a},f_{z/a})(au)\,du 
\\
& = - e^{ -\frac{|z|^2+|w|^2+z^2+w^2}{a^2} } \,\frac{1}{2a\sqrt{\pi}}\Big( \frac{1}{z}-\frac{1}{w}    \Big) \int_{-1}^1  \, e^{4u(z+w)} \,du \,(1+o(1))
\\
& = - e^{ -\frac{|z|^2+|w|^2+z^2+w^2}{a^2} } \,\frac{1}{8a\sqrt{\pi}} \frac{1}{zw} \frac{w-z}{z+w} \Big( e^{4(z+w)}-e^{-4(z+w)} \Big) \,(1+o(1)) \quad \mbox{as } a \to 0.
\end{split}
\end{equation}
Summing the last asymptotic formula with \eqref{lol7} (and using \eqref{kappa Wronskian}) gives
\begin{align}
& \wt{\kappa}(z,w):=	\frac{1}{a^3} e^{-\frac{|z|^2+|w|^2}{a^2}} \kappa^{\mathbb{R}}(\frac{z}{a},\frac{w}{a})  \nonumber  \\
 &=  \frac{2}{a\sqrt{\pi}} \,e^{-\frac{|z|^2+|w|^2+z^2+w^2}{a^2}} \,  \frac{z-w}{4zw} \Big( \frac{ \sinh(4(z+w)) }{2(z+w)} - \frac{  \sinh(4(z-w))  }{2(z-w)}  +o(1) \Big), \quad \mbox{as } a \to 0. \label{kappa wt asymp a 0}
\end{align}
We now use the Gaussian approximation of the Dirac delta: for any continuous function $f:\mathbb{R}\to \mathbb{C}$ with compact support and any fixed $\lambda>0$, we have
\begin{equation} \label{Gaussian Dirac}
\int_{-\infty}^{+\infty}\frac{1}{a \sqrt{\pi}} e^{-\lambda(x/a)^2}f(x)dx = \frac{f(0)}{\sqrt{\lambda}} +o(1) \qquad  \mbox{as } a \to 0.
\end{equation}
For short, in what follows we will denote the above as $\frac{1}{a \sqrt{\pi}} e^{-\lambda(x/a)^2} \overset{d}{\longrightarrow} \frac{\delta(x)}{\sqrt{\lambda}}$ as $a \to 0$. 

Before considering the general case $k \in \mathbb{N}_{>0}$, it is instructive to first look at the simpler cases $k=1,2$.
For $k=1$, by \eqref{lol9}, \eqref{kappa wt asymp a 0} and \eqref{Gaussian Dirac}, 
\begin{equation}
\begin{split}
\wt{R}_1(z_1)  & =\wt{\kappa}(z_1,\bar{z}_1) (-2i y_1)=  \frac{1}{a\sqrt{\pi}} \,e^{-4(x_1/a)^2 } \frac{y_1^2}{x_1^2+y_1^2}
\Big( \frac{ \sinh(8x_1) }{2x_1} - \frac{  \sin(8y_1)  }{2y_1}  +o(1) \Big)
\\
& \overset{d}{\longrightarrow}  \frac{ \delta(x_1) }{ 2 }  \frac{y_1^2}{x_1^2+y_1^2}
\Big( \frac{ \sinh(8x_1) }{2x_1} - \frac{  \sin(8y_1)  }{2y_1}  +o(1) \Big) = \delta(x_1) \,  K^{\mathrm{\chi sin}}(y_1,y_1), \qquad  \mbox{as } a \to 0.
\end{split}
\end{equation}
For $k=2$, using an exact computation of the Pfaffian in \eqref{lol9}, we get
\begin{equation}
\begin{split}
\wt{R}_2(z_1,z_2)= \Big( \wt{\kappa}(z_1,\bar{z}_1)\wt{\kappa}(z_2,\bar{z}_2)-|\wt{\kappa}(z_1,z_2)|^2+|\wt{\kappa}(z_1,\bar{z}_2)|^2 \Big) \, 4y_1y_2.
\end{split}
\end{equation}
By \eqref{kappa wt asymp a 0}, as $a \to 0$, the terms in the above right-hand side have the same exponential factor $\frac{4}{a^2\pi} e^{-4(x_1/a)^2-4(x_2/a)^2}$ in their asymptotics. Using \eqref{Gaussian Dirac}, we obtain after some computation that
\begin{equation}
\wt{R}_2(z_{1},z_{2})   \overset{d}{\longrightarrow}  \delta (x_1) \delta(x_2) \Big( K^{\mathrm{ \chi sin}}(y_1,y_1) K^{\mathrm{ \chi sin}}(y_2,y_2)- K^{\mathrm{ \chi sin}}(y_1,y_2)^2  \Big),  \qquad  \mbox{as } a \to 0.
\end{equation}

We now turn to the general case. 
For $k \in \mathbb{N}_{>0}$, using \eqref{kappa wt asymp a 0} and 
\begin{equation}
	\prod_{j=1}^k \frac{2}{a\sqrt{\pi}} e^{-4(x_j/a)^2}    \overset{d}{\longrightarrow} 	\prod_{j=1}^k \delta (x_j), 
\end{equation}
we obtain
\begin{align}
& \wt{R}_k(z_1,\dots,z_k)  = \Pf \Big[ 
\begin{pmatrix} 
\wt{\kappa}(z_j,z_l) & \wt{\kappa}(z_j,\bar{z}_l) \smallskip \\
\wt{\kappa}(\bar{z}_j,z_l) & \wt{\kappa}(\bar{z}_j,\bar{z}_l) 
\end{pmatrix} 
\Big]_{j,l=1}^k \prod_{j=1}^k (-2i \, y_j) \\
& \overset{d}{\longrightarrow}   \Pf \Big[  \begin{pmatrix}
\frac{y_j-y_l}{2y_j y_l} K^{\mathrm{\chi sin}}(y_j,y_l) & \frac{y_j+y_l}{2y_j y_l} K^{\mathrm{\chi sin}}(y_j,y_l) \smallskip \\
-\frac{y_j+y_l}{2y_j y_l} K^{\mathrm{\chi sin}}(y_j,y_l) & -\frac{y_j-y_l}{2y_j y_l} K^{\mathrm{\chi sin}}(y_j,y_l)
\end{pmatrix} \Big]_{j,l=1}^k  \prod_{j=1}^k   y_j\, \delta (x_j),  \qquad  \mbox{as } a \to 0.
\end{align}
Now Lemma~\ref{Lem_Pf det iden} completes the proof. 
\end{proof}

\section{Proof of Theorem \ref{Thm_partition functions}: semi-large gap probabilities} \label{Section_Parition functions}

The first step in proving Theorem \ref{Thm_partition functions} is to obtain exact identities for $\mathbb{P}_{N}^{1}$, $\mathbb{P}_{N}^{2}$ and $\mathbb{P}_{N}^{12}$. For this, we will rely on the following well-known formula for partition functions of planar symplectic ensembles. 

\begin{lem}(See e.g. \cite[Remark 2.5 and Corollary 3.3]{akemann2021skew})\label{lemma:partition exact}

\noindent Let $\mathsf{w}$ be a rotation-invariant weight with sufficient decay at $\infty$,
\begin{align*}
\int_{0}^{+\infty} r^{j}\mathsf{w}(r)dr < + \infty, \qquad \mbox{for all } j \geq 0.
\end{align*}
Then the partition function
\begin{align}\label{tilde ZN}
\tilde{Z}_{N} := \frac{1}{N!} \int_{\C}\ldots \int_{\C} \prod_{1 \leq j<k \leq N} \abs{\zeta_j-\zeta_k}^2 \abs{\zeta_j-\overline{\zeta}_k}^2 \prod_{j=1}^{N} \abs{\zeta_j-\overline{\zeta}_j}^2 \mathsf{w}(\zeta_j)  \,  dA(\zeta_j)
\end{align}
can be rewritten as
\begin{equation} \label{ZN symplectic}
\tilde{Z}_N =  \prod_{k=0}^{N-1} \tilde{r}_k, \qquad \mbox{where } \quad \tilde{r}_{k} := 2\int_\C |\zeta|^{4k+2} \mathsf{w}(\zeta) \, dA(\zeta) = 4\int_{0}^{+\infty} r^{4k+3}\mathsf{w}(r)dr.
\end{equation}
\end{lem}

\begin{lem} \label{Lem_ZN repre}
For any $\rho >0$ and $N \in \mathbb{N}_{>0}$, the following identities hold:
\begin{align}
 \log \mathbb{P}_{N}^{1} &= \sum_{j=0}^{N-1} \log \Big( 1-  \frac{\gamma(2+2j+b_N N,b_NN )}{\Gamma(2+2j+b_N N)} \Big),  \label{PN1 gamma}
 \\
 \log \mathbb{P}_{N}^{2} & = \sum_{j=0}^{N-1} \log \Big(   \frac{\gamma(2+2j+b_N N,b_NN +2N)}{\Gamma(2+2j+b_N N)} \Big), \label{PN2 gamma}
 \\
 \log \mathbb{P}_{N}^{12} & = \sum_{j=0}^{N-1} \log \Big(  \frac{\gamma(2+2j+b_N N,b_NN +2N)}{\Gamma(2+2j+b_N N)}-\frac{\gamma(2+2j+b_N N,b_NN )}{\Gamma(2+2j+b_N N)} \Big),  \label{PN12 gamma}
\end{align}
where $\gamma$ is defined in \eqref{incom Gamma}.
\end{lem}

\begin{proof}
Let 
\begin{align*}
& \mathsf{w}_{1}(\zeta) := e^{-NQ_{N}(\zeta)} \begin{cases}
0, & \mbox{if } |\zeta| \leq r_{1}, \\
1, & \mbox{otherwise},
\end{cases} & & \mathsf{w}_{2}(\zeta) := e^{-NQ_{N}(\zeta)} \begin{cases}
0, & \mbox{if } |\zeta| \geq r_{2}, \\
1, & \mbox{otherwise},
\end{cases} \\
& \mathsf{w}_{12}(\zeta) := e^{-NQ_{N}(\zeta)} \begin{cases}
0, & \mbox{if } |\zeta| \in [0,r_{1}]\cup[r_{2},\infty], \\
1, & \mbox{otherwise}.
\end{cases} 
\end{align*}
By the definitions \eqref{PN 1}--\eqref{PN 12}, we have $\mathbb{P}_{N}^{1}=Z_{N}^{1}/Z_{N}$, $\mathbb{P}_{N}^{2}=Z_{N}^{2}/Z_{N}$ and $\mathbb{P}_{N}^{12}=Z_{N}^{12}/Z_{N}$, where $Z_{N}$, $Z_{N}^{1}$, $Z_{N}^{2}$, $Z_{N}^{12}$ are given by the right-hand side of \eqref{tilde ZN} with $\mathsf{w}$ replaced by $e^{-NQ_{N}}$, $\mathsf{w}_{1}$, $\mathsf{w}_{2}$, $\mathsf{w}_{12}$, respectively. Combining Lemma \ref{lemma:partition exact} with
\begin{align}
& 2 \int_{r_1}^\infty r^{4k+3+2b_NN} e^{-a_N N r^2}\,dr=\frac{ \Gamma(2+2k+b_N N,a_NN \,r_1^2)  } { (a_N N)^{2+2k+b_N N}  },
\\
& 2 \int_{0}^{r_2} r^{4k+3+2b_NN} e^{-a_N N r^2}\,dr=\frac{ \gamma(2+2k+b_N N,a_NN \,r_2^2)  } { (a_N N)^{2+2k+b_N N}  },
\\
&2 \int_{r_1}^{r_2} r^{4k+3+2b_NN} e^{-a_N N r^2}\,dr=\frac{ \gamma(2+2k+b_N N,a_NN \,r_2^2)-\gamma(2+2k+b_N N,a_NN \,r_1^2)  } { (a_N N)^{2+2k+b_N N}  },
\end{align}
we obtain
\begin{align}
& Z_N= \frac{2^N}{(a_NN)^{(b_N+1)N^2+N}} \prod_{k=0}^{N-1}  \Gamma(2+2k+b_N N),   \label{ZN rep}
\\
& Z_N^1=  \frac{2^N}{(a_NN)^{(b_N+1)N^2+N}} \prod_{k=0}^{N-1}  \Gamma(2+2k+b_N N,b_NN ) ,  \label{ZN 1 rep}
\\
& Z_N^2 =  \frac{2^N}{(a_NN)^{(b_N+1)N^2+N}} \prod_{k=0}^{N-1}  \gamma(2+2k+b_N N, b_NN+2N ) ,  \label{ZN 2 rep}
\\
& Z_N^{12} =  \frac{2^N}{(a_NN)^{(b_N+1)N^2+N}} \prod_{k=0}^{N-1} \Big( \gamma(2+2k+b_N N, b_NN+2N )-\gamma(2+2k+b_N N, b_NN )  \Big),   \label{ZN 12 rep}
\end{align}
and the claim follows.
\end{proof}

\begin{proof}[Proof of Theorem~\ref{Thm_partition functions}]
By Lemmas \ref{Lem_ZN repre} and \ref{lemma: uniform}, we have
\begin{align*}
\log \mathbb{P}_{N}^{1} & = \sum_{j=0}^{N-1} \log \Big( 1-  \frac{1}{2}\erfc\Big(-\eta_{j,1} \sqrt{\tilde{a}_{j}/2}\Big) + R_{\tilde{a}_{j}}(\eta_{j,1}) \Big),  
\\
\log \mathbb{P}_{N}^{2} & = \sum_{j=0}^{N-1} \log \Big(   \frac{1}{2}\erfc\Big(-\eta_{j,2} \sqrt{\tilde{a}_{j}/2}\Big) - R_{\tilde{a}_{j}}(\eta_{j,2}) \Big), 
\\
\log \mathbb{P}_{N}^{12} & = \sum_{j=0}^{N-1} \log \Big(  \frac{1}{2}\erfc\Big(-\eta_{j,2} \sqrt{\tilde{a}_{j}/2}\Big) - R_{\tilde{a}_{j}}(\eta_{j,2}) -  \frac{1}{2}\erfc\Big(-\eta_{j,1} \sqrt{\tilde{a}_{j}/2}\Big) + R_{\tilde{a}_{j}}(\eta_{j,1}) \Big),
\end{align*}
where 
\begin{align}
& \tilde{a}_{j} = 2+2j+b_{N} N, \qquad \lambda_{j,1}=\frac{b_{N} N}{\tilde{a}_{j}}, \qquad \lambda_{j,2} = \frac{b_{N}N + 2N}{\tilde{a}_{j}}, \label{Aj lambda j12} \\
& \eta_{j,k} = (\lambda_{j,k}-1)\sqrt{\frac{2(\lambda_{j,k}-1-\log \lambda_{j,k})}{(\lambda_{j,k}-1)^{2}}}, \qquad k=1,2. \label{eta jk}
\end{align}
We complete the proof and give full details only for $\mathbb{P}_{N}^{12}(\rho)$. The proofs for $\mathbb{P}_{N}^{1}(\rho)$ and $\mathbb{P}_{N}^{2}(\rho)$ are similar (and simpler), so we omit them. 
Using \eqref{Q aN bN}, \eqref{Aj lambda j12} and \eqref{eta jk}, we find
\begin{align}
& \lambda_{j,1}=1-\frac{2\rho^2(1+j)}{N^2}+O(N^{-3}), & & \lambda_{j,2}= 1+\frac{2\rho^2}{N}-\frac{2\rho^2(1+j-\rho^2)}{N^2}+O(N^{-3}),
\\
& \eta_{j,1}= -\frac{2(1+j)\rho^2}{N^2}+O(N^{-3}), & & \eta_{j,2}=\frac{2\rho^2}{N}-\frac{2\rho^2(3+3j+\rho^2)}{3N^2}+O(N^{-3}),
\end{align} 
as $N \to \infty$. It follows that
\begin{equation}
\begin{split}
-\eta_{j,1} \sqrt{\tilde{a}_{j}/2} &=\sqrt{2}\rho \frac{1+j}{N} +\frac{(1+j)\rho^3}{\sqrt{2}N^2}+O(N^{-3}), 
\\
-\eta_{j,2} \sqrt{\tilde{a}_{j}/2} &= -\sqrt{2}\rho \Big(1-\frac{j}{N}\Big) +\sqrt{2} \rho \frac{ 1+\rho^2/6 }{N}+\frac{ \rho^3(12+12j-11\rho^2) }{ 36\sqrt{2} N^2 }+O(N^{-3}),
\end{split}
\end{equation}
as $N \to \infty$. Using the above and \eqref{asymp of Ra}, we then get
\begin{equation}
\begin{split}
&\quad \log \Big(  \frac{1}{2}\erfc\Big(-\eta_{j,2} \sqrt{\tilde{a}_{j}/2}\Big) - R_{\tilde{a}_{j}}(\eta_{j,2}) -  \frac{1}{2}\erfc\Big(-\eta_{j,1} \sqrt{\tilde{a}_{j}/2}\Big) + R_{\tilde{a}_{j}}(\eta_{j,1}) \Big)
\\
&= f_{1}(j/N) + \frac{f_{2}(j/N)}{N} + O(N^{-2}),
\end{split}
\end{equation}
uniformly for $j \in \{0,1,\ldots,N-1\}$, where
\begin{align*}
f_{1}(x) &:= \log \Big( \frac{1}{2}\mathrm{erfc}(-\sqrt{2} \, \rho (1-x)) - \frac{1}{2}\mathrm{erfc}(\sqrt{2} \, \rho x) \Big), \\
f_{2}(x) &:=  \frac{-\rho}{3\sqrt{2\pi}}\frac{e^{-2 \rho^{2} x^{2}}(2x^{2} \rho^{2} - 3 \rho^{2} x - 5)+e^{-2\rho^{2}(1-x)^{2}}(5+\rho^{2}(1+x-2x^{2}))}{\frac{1}{2}\mathrm{erfc}(-\sqrt{2} \, \rho (1-x)) - \frac{1}{2}\mathrm{erfc}(\sqrt{2} \, \rho x)}.
\end{align*}
Since $\{f_{1}^{(\ell)},f_{2}^{(\ell)}\}_{\ell=0,1,2}$ are continuous and bounded on $[0,1]$, it follows from \cite[Lemma 3.4]{charlier2021large} (with $A=a_{0}=0$, $B=1$ and $b_{0}=-1$) that as $N \to + \infty$,
\begin{align*}
&\sum_{j=0}^{N-1}f_{1}(j/N) = N \int_{0}^{1}f_{1}(x) \, dx + \frac{f_{1}(0)-f_{1}(1)}{2} + O(N^{-1}), 
\\
&\frac{1}{N}\sum_{j=0}^{N-1}f_{2}(j/N) = \int_{0}^{1}f_{2}(x) \, dx + O(N^{-1}).
\end{align*}
This yields
\begin{align*}
\log \Prob_N^{12}(\rho) &= N \int_{0}^{1} \log \Big( \frac{1}{2}\mathrm{erfc}(-\sqrt{2}\, \rho (1-x))-\frac{1}{2}\mathrm{erfc}(\sqrt{2}\, \rho x) \Big)\, dx + \frac{1}{2}\log \Big( \frac{\mathrm{erfc}(-\sqrt{2}\, \rho)-1}{1-\mathrm{erfc}(\sqrt{2}\, \rho)} \Big)
\\
&+\frac{2\rho}{3\sqrt{2\pi}}  \int_{0}^{1}  \frac{ e^{-2\rho^{2}x^{2}}(5+3\rho^{2} x - 2\rho^{2}x^{2}) - e^{-2\rho^{2}(1-x)^{2}}(5+(1+x-2x^{2})\rho^{2})}{ \mathrm{erfc}(-\sqrt{2}\, \rho (1-x))-\mathrm{erfc}(\sqrt{2}\, \rho x) } \, dx+ O(N^{-1})
\end{align*}
as $N \to + \infty$. Since $\erfc(-z)=2-\erfc(z)$, 
\begin{equation}
\log \Big( \frac{\mathrm{erfc}(-\sqrt{2}\, \rho)-1}{1-\mathrm{erfc}(\sqrt{2}\, \rho)} \Big) = 0. 
\end{equation}
Using also
\begin{equation}
\int_{0}^{1} \frac{ e^{-2\rho^{2}x^{2}}(5+3\rho^{2} x - 2\rho^{2}x^{2}) }{ \mathrm{erfc}(-\sqrt{2}\, \rho (1-x))-\mathrm{erfc}(\sqrt{2}\, \rho x) } \, dx=  \int_{0}^{1}  \frac{ e^{-2\rho^{2}(1-x)^{2}}(5+(1+x-2x^{2})\rho^{2})}{ \mathrm{erfc}(-\sqrt{2}\, \rho (1-x))-\mathrm{erfc}(\sqrt{2}\, \rho x) } \, dx. 
\end{equation}
we obtain \eqref{PN 12 asym thm}. 
\end{proof}

\appendix

\section{Uniform asymptotics of the incomplete gamma function}\label{appendix_gamma}

In this appendix, we collect some known asymptotics of the incomplete gamma function.

\begin{lem}(Taken from \cite[Section 8.11.7]{olver2010nist}, \cite{MR45253} and \cite[(1.32) and below]{ameur2021szego}). \label{Lemma:ameur gamma} Let $E_{\mathrm{sz}}$ be the exterior region of the Szeg\H o curve $\{ z\in \C: |z| \le 1, |z e^{1-z}|=1 \}$. Note that $\{z:|\arg(z-1)|<\frac{3\pi}{4}\} \subset E_{\mathrm{sz}}$. As $\tilde{a} \to +\infty$,  
\begin{equation} \label{Q asymp outer gen}
\mathrm{Q}( \tilde{a} , \tilde{a} z) \sim \frac{ \tilde{a}^{\tilde{a}-1} }{ \Gamma(\tilde{a}) } e^{-\tilde{a}z} \frac{z^{\tilde{a}}}{z-1} \sum_{k=0}^{+\infty} \frac{b_k(z)}{ (z-1)^{2k+1} } \frac{ (-1)^k }{ \tilde{a}^k }
\end{equation}
uniformly for $z$ in compact susbets of $E_{\mathrm{sz}}$. The coefficients $b_k(z)$ are defined recursively by 
\begin{equation}
b_0(z)=1, \qquad b_k(z)=z(1-z)b'_{k-1}(z)+(2k-1) zb_{k-1}(z), \quad (k=1,2,\dots). 
\end{equation}
In particular, as $\tilde{a}\to + \infty$ we have 
\begin{equation} \label{Q asymp outer}
	\mathrm{Q}( \tilde{a} , \tilde{a} z) = \frac{1}{\sqrt{2\pi \tilde{a}}}\, e^{\tilde{a}-\tilde{a} z} \frac{z^{\tilde{a}} }{z-1}\Big(1-\Big(\frac{1}{12}+\frac{z}{(z-1)^2}\Big)\frac{1}{\tilde{a}} +\Big(\frac{1}{288}+\frac{z(2z+1)}{(z-1)^4}\Big) \frac{1}{\tilde{a}^2}+O(\frac{1}{\tilde{a}^3}) \Big),
\end{equation}
uniformly for $z$ in compact susbets of $E_{\mathrm{sz}}$. 	
\end{lem}
The asymptotic expansion \eqref{Q asymp outer gen} is stated in \cite{MR45253} for $|\arg(z-1)|<\frac{3\pi}{4}$ and extended in \cite[(1.32) and below]{ameur2021szego} for $z \in E_{\mathrm{sz}}$. Lemma \ref{Lemma:ameur gamma} is used in Subsection \ref{Subsec_cplex limit}, see \eqref{ameur cron}. In fact, to prove Theorem~\ref{thm:main} (a) for a given $\theta \in [0,2\pi)\setminus \{0,\pi\}$, we need \eqref{Q asymp outer} with $z$ close to $e^{2i\theta}$. Hence, to handle the case $\theta \notin [\frac{\pi}{4},\frac{3\pi}{4}]\cup [\frac{5\pi}{4},\frac{7\pi}{4}]$, the results from \cite{MR45253} are enough for us; however for the other case $\theta \in [\frac{\pi}{4},\frac{3\pi}{4}]\cup [\frac{5\pi}{4},\frac{7\pi}{4}]$ we rely on \cite{ameur2021szego}.

\begin{lemma}\label{lemma: uniform}(Taken from \cite[Section 11.2.4]{Temme}).
	For $\tilde{a}>0$ and $z>0$, we have
	\begin{align*}
		& \frac{\gamma(\tilde{a},z)}{\Gamma(\tilde{a})} = \frac{1}{2}\mathrm{erfc}(-\eta \sqrt{\tilde{a}/2}) - R_{\tilde{a}}(\eta), \qquad R_{\tilde{a}}(\eta) = \frac{e^{-\frac{1}{2}\tilde{a} \eta^{2}}}{2\pi i}\int_{-\infty}^{\infty}e^{-\frac{1}{2}\tilde{a} u^{2}}g(u)du,
	\end{align*}
	where $g(u) = \frac{dt}{du}\frac{1}{\lambda-t}+\frac{1}{u+i \eta}$,
	\begin{align}\label{lol8}
		& \lambda = \frac{z}{\tilde{a}}, \quad \eta = (\lambda-1)\sqrt{\frac{2 (\lambda-1-\log \lambda)}{(\lambda-1)^{2}}}, \quad u = -i(t-1)\sqrt{\frac{2(t-1-\log t)}{(t-1)^{2}}},
	\end{align}
	where $\mathrm{sign} (\eta) = \mathrm{sign}(\lambda-1)$, and $\mathrm{sign}(u) = \mathrm{sign}(\im t)$ with $t \in \mathcal{L}:=\{\frac{\theta}{\sin \theta} e^{i\theta}: -\pi < \theta < \pi\}$ and $u \in \mathbb{R}$ (in particular $u = -i(t-1)+\bigO((t-1)^{2})$ as $t \to 1$). Furthermore, 
	\begin{align}\label{asymp of Ra}
		& R_{\tilde{a}}(\eta) \sim \frac{e^{-\frac{1}{2}\tilde{a} \eta^{2}}}{\sqrt{2\pi \tilde{a}}}\sum_{j=0}^{+\infty} \frac{c_{j}(\eta)}{\tilde{a}^{j}} \qquad \mbox{as } \tilde{a} \to + \infty
	\end{align}
	uniformly for $z \in [0,\infty)$, where all coefficients $c_{j}(\eta)$ are bounded functions of $\eta \in \mathbb{R}$ (i.e. bounded for $\lambda \in [0,\infty)$) and given by
	\begin{align}\label{recursive def of the cj}
		c_{0} = \frac{1}{\lambda-1} - \frac{1}{\eta}, \qquad c_{j} = \frac{1}{\eta} \frac{d}{d\eta}c_{j-1}(\eta) + \frac{\gamma_{j}}{\lambda-1}, \; j \geq 1,
	\end{align}
	where the $\gamma_{j}$ are the Stirling coefficients
	\begin{align*}
		\gamma_{j} = \frac{(-1)^{j}}{2^{j} \, j!} \bigg[ \frac{d^{2j}}{dx^{2j}} \bigg( \frac{1}{2}\frac{x^{2}}{x-\log(1+x)} \bigg)^{j+\frac{1}{2}} \bigg]_{x=0}.
	\end{align*}
\end{lemma}

\subsection*{Acknowledgements} We are grateful to Gernot Akemann and Peter Forrester for their interest and helpful discussions. SB acknowledges support from the National Research Foundation of Korea, Grant NRF-2019R1A5A1028324, and Samsung Science and Technology Foundation, 	Grant SSTF-BA1401-51. CC acknowledges support from the Novo Nordisk Fonden Project, Grant 0064428, the Swedish Research Council, Grant No. 2021-04626, and the Ruth and Nils-Erik Stenb\"{a}ck Foundation.

\bibliographystyle{abbrv}
\bibliography{RMTbib}

\begin{thebibliography}{10}

\bibitem{MR3814242}
K.~Adhikari.
\newblock Hole probabilities for {$\beta$}-ensembles and determinantal point
  processes in the complex plane.
\newblock {\em Electron. J. Probab.}, 23:Paper No. 48, 21, 2018.

\bibitem{MR3719476}
K.~Adhikari and N.~K. Reddy.
\newblock Hole probabilities for finite and infinite {G}inibre ensembles.
\newblock {\em Int. Math. Res. Not. IMRN}, (21):6694--6730, 2017.

\bibitem{MR1762659}
M.~Adler, P.~J. Forrester, T.~Nagao, and P.~van Moerbeke.
\newblock Classical skew orthogonal polynomials and random matrices.
\newblock {\em J. Statist. Phys.}, 99(1-2):141--170, 2000.

\bibitem{MR2180006}
G.~Akemann.
\newblock The complex {L}aguerre symplectic ensemble of non-{H}ermitian
  matrices.
\newblock {\em Nuclear Phys. B}, 730(3):253--299, 2005.

\bibitem{MR2302902}
G.~Akemann and F.~Basile.
\newblock Massive partition functions and complex eigenvalue correlations in
  matrix models with symplectic symmetry.
\newblock {\em Nuclear Phys. B}, 766(1-3):150--177, 2007.

\bibitem{akemann2021scaling}
G.~Akemann, S.-S. Byun, and N.-G. Kang.
\newblock Scaling limits of planar symplectic ensembles.
\newblock {\em SIGMA Symmetry Integrability Geom. Methods Appl.}, 18:Paper No.
  007, 40, 2022.

\bibitem{akemann2016universality}
G.~Akemann, M.~Cikovic, and M.~Venker.
\newblock Universality at weak and strong non-{H}ermiticity beyond the elliptic
  {G}inibre ensemble.
\newblock {\em Comm. Math. Phys.}, 10.1007/s00220-018-3201-1, 2018.

\bibitem{akemann2021skew}
G.~Akemann, M.~Ebke, and I.~Parra.
\newblock Skew-orthogonal polynomials in the complex plane and their
  {B}ergman-like kernels.
\newblock {\em Comm. Math. Phys.}, 389:621--659, 2022.

\bibitem{MR3279619}
G.~Akemann, J.~R. Ipsen, and E.~Strahov.
\newblock Permanental processes from products of complex and quaternionic
  induced {G}inibre ensembles.
\newblock {\em Random Matrices Theory Appl.}, 3(4):1450014, 54, 2014.

\bibitem{akemann2019universal}
G.~Akemann, M.~Kieburg, A.~Mielke, and T.~Prosen.
\newblock Universal signature from integrability to chaos in dissipative open
  quantum systems.
\newblock {\em Phys. Rev. Lett.}, 123(25):254101, 2019.

\bibitem{MR3192169}
G.~Akemann and M.~J. Phillips.
\newblock The interpolating {A}iry kernels for the {$\beta=1$} and {$\beta=4$}
  elliptic {G}inibre ensembles.
\newblock {\em J. Stat. Phys.}, 155(3):421--465, 2014.

\bibitem{MR2536111}
G.~Akemann, M.~J. Phillips, and L.~Shifrin.
\newblock Gap probabilities in non-{H}ermitian random matrix theory.
\newblock {\em J. Math. Phys.}, 50(6):063504, 32, 2009.

\bibitem{MR3063493}
G.~Akemann and E.~Strahov.
\newblock Hole probabilities and overcrowding estimates for products of complex
  {G}aussian matrices.
\newblock {\em J. Stat. Phys.}, 151(6):987--1003, 2013.

\bibitem{AB21}
Y.~Ameur and S.-S. Byun.
\newblock Almost-{H}ermitian random matrices and bandlimited point processes.
\newblock {\em preprint arXiv:2101.03832}, 2021.

\bibitem{ameur2021szego}
Y.~Ameur and J.~Cronvall.
\newblock Szeg\"{o} type asymptotics for the reproducing kernel in spaces of
  full-plane weighted polynomials.
\newblock {\em preprint arXiv:2107.11148}, 2021.

\bibitem{MR4030288}
Y.~Ameur, N.-G. Kang, N.~Makarov, and A.~Wennman.
\newblock Scaling limits of random normal matrix processes at singular boundary
  points.
\newblock {\em J. Funct. Anal.}, 278(3):108340, 2020.

\bibitem{MR2934715}
F.~Benaych-Georges and F.~Chapon.
\newblock Random right eigenvalues of {G}aussian quaternionic matrices.
\newblock {\em Random Matrices Theory Appl.}, 1(2):1150009, 18, 2012.

\bibitem{byun2021universal}
S.-S. Byun and M.~Ebke.
\newblock Universal scaling limits of the symplectic elliptic {G}inibre
  ensembles.
\newblock {\em preprint arXiv:2108.05541}, 2021.

\bibitem{byun2021wronskian}
S.-S. Byun, M.~Ebke, and S.-M. Seo.
\newblock Wronskian structures of planar symplectic ensembles.
\newblock {\em preprint arXiv:2110.12196}, 2021.

\bibitem{byun2021random}
S.-S. Byun and S.-M. Seo.
\newblock Random normal matrices in the almost-circular regime.
\newblock {\em Bernoulli (to appear), arXiv:2112.11353}, 2021.

\bibitem{charlier2021large}
C.~Charlier.
\newblock Large gap asymptotics on annuli in the random normal matrix model.
\newblock {\em preprint arXiv:2110.06908}, 2021.

\bibitem{MR2881072}
J.~Fischmann, W.~Bruzda, B.~A. Khoruzhenko, H.-J. Sommers, and
  K.~\.{Z}yczkowski.
\newblock Induced {G}inibre ensemble of random matrices and quantum operations.
\newblock {\em J. Phys. A}, 45(7):075203, 31, 2012.

\bibitem{MR1181356}
P.~J. Forrester.
\newblock Some statistical properties of the eigenvalues of complex random
  matrices.
\newblock {\em Phys. Lett. A}, 169(1-2):21--24, 1992.

\bibitem{forrester2010log}
P.~J. Forrester.
\newblock {\em Log-gases and {R}andom {M}atrices (LMS-34)}.
\newblock Princeton University Press, Princeton, 2010.

\bibitem{F2014}
P.~J. Forrester.
\newblock Asymptotics of spacing distributions 50 years later.
\newblock {\em Random matrix theory, interacting particle systems, and
  integrable systems, Math. Sci. Res. Inst. Publ}, 65:199--222, 2014.

\bibitem{forrester2016analogies}
P.~J. Forrester.
\newblock Analogies between random matrix ensembles and the one-component
  plasma in two-dimensions.
\newblock {\em Nucl. Phys. B}, 904:253--281, 2016.

\bibitem{fyodorov1997almost}
Y.~V. Fyodorov, B.~A. Khoruzhenko, and H.-J. Sommers.
\newblock Almost {H}ermitian random matrices: crossover from {W}igner-{D}yson
  to {G}inibre eigenvalue statistics.
\newblock {\em Phys. Rev. Lett.}, 79(4):557--560, 1997.

\bibitem{MR1431718}
Y.~V. Fyodorov, B.~A. Khoruzhenko, and H.-J. Sommers.
\newblock Almost-{H}ermitian random matrices: eigenvalue density in the complex
  plane.
\newblock {\em Phys. Lett. A}, 226(1-2):46--52, 1997.

\bibitem{MR1634312}
Y.~V. Fyodorov, H.-J. Sommers, and B.~A. Khoruzhenko.
\newblock Universality in the random matrix spectra in the regime of weak
  non-{H}ermiticity.
\newblock {\em Ann. Inst. H. Poincar\'{e} Phys. Th\'{e}or.}, 68(4):449--489,
  1998.

\bibitem{GN2018}
S.~Ghosh and A.~Nishry.
\newblock Point processes, hole events, and large deviations: random complex
  zeros and {C}oulomb gases.
\newblock {\em Constr. Approx.}, 48(1):101--136, 2018.

\bibitem{ginibre1965statistical}
J.~Ginibre.
\newblock Statistical ensembles of complex, quaternion, and real matrices.
\newblock {\em J. Math. Phys.}, 6(3):440--449, 1965.

\bibitem{GHS1988}
R.~Grobe, F.~Haake, and H.-J. Sommers.
\newblock Quantum distinction of regular and chaotic dissipative motion.
\newblock {\em Phys. Rev. Lett.}, 61(17):1899, 1988.

\bibitem{MR3066113}
J.~R. Ipsen.
\newblock Products of independent quaternion {G}inibre matrices and their
  correlation functions.
\newblock {\em J. Phys. A}, 46(26):265201, 16, 2013.

\bibitem{MR1239571}
B.~Jancovici, J.~L. Lebowitz, and G.~Manificat.
\newblock Large charge fluctuations in classical {C}oulomb systems.
\newblock {\em J. Statist. Phys.}, 72(3-4):773--787, 1993.

\bibitem{MR1928853}
E.~Kanzieper.
\newblock Eigenvalue correlations in non-{H}ermitean symplectic random
  matrices.
\newblock {\em J. Phys. A}, 35(31):6631--6644, 2002.

\bibitem{khoruzhenko2021truncations}
B.~A. Khoruzhenko and S.~Lysychkin.
\newblock Truncations of random symplectic unitary matrices.
\newblock {\em preprint arXiv:2111.02381}, 2021.

\bibitem{kiessling1999note}
M.~K.-H. Kiessling and H.~Spohn.
\newblock A note on the eigenvalue density of random matrices.
\newblock {\em Comm. Math. Phys.}, 199(3):683--695, 1999.

\bibitem{kuijlaars2011universality}
A.~B.~J. Kuijlaars.
\newblock Universality.
\newblock {\em Chapter 6 in: Oxford Handbook of Random Matrix Theory, (G.
  Akemann, J. Baik, and P. Di Francesco, eds.), Oxford University Press, 2011,
  arXiv:1103.5922}, 2011.

\bibitem{L2019}
B.~Lacroix-A-Chez-Toine, J.~A.~M. Garz{\'o}n, C.~S.~H. Calva, I.~P. Castillo,
  A.~Kundu, S.~N. Majumdar, and G.~Schehr.
\newblock Intermediate deviation regime for the full eigenvalue statistics in
  the complex {G}inibre ensemble.
\newblock {\em Phys. Rev. E}, 100(1):012137, 2019.

\bibitem{MR3450566}
S.-Y. Lee and R.~Riser.
\newblock Fine asymptotic behavior for eigenvalues of random normal matrices:
  ellipse case.
\newblock {\em J. Math. Phys.}, 57(2):023302, 29, 2016.

\bibitem{Lysychkin}
S.~Lysychkin.
\newblock Complex eigenvalues of high dimensional quaternion random matrices.
\newblock {\em PhD Thesis, Queen Mary University of London United Kingdom},
  2021.

\bibitem{MR3612266}
A.~Mays and A.~Ponsaing.
\newblock An induced real quaternion spherical ensemble of random matrices.
\newblock {\em Random Matrices Theory Appl.}, 6(1):1750001, 29, 2017.

\bibitem{Mehta}
M.~L. Mehta.
\newblock {\em Random matrices}.
\newblock Academic Press, Inc., Boston, MA, second edition, 1991.

\bibitem{olver2010nist}
F.~W. Olver, D.~W. Lozier, R.~F. Boisvert, and C.~W. Clark~(Editors).
\newblock {\em NIST Handbook of Mathematical Functions}.
\newblock Cambridge University Press, Cambridge, 2010.

\bibitem{ST97}
E.~B. Saff and V.~Totik.
\newblock {\em Logarithmic potentials with external fields}, volume 316 of {\em
  Grundlehren der Mathematischen Wissenschaften [Fundamental Principles of
  Mathematical Sciences]}.
\newblock Springer-Verlag, Berlin, 1997.
\newblock Appendix B by Thomas Bloom.

\bibitem{JacEdward}
E.~V. Shuryak and J.~Verbaarschot.
\newblock Random matrix theory and spectral sum rules for the dirac operator in
  {QCD}.
\newblock {\em Nucl. Phys. A}, 560(1):306--320, 1993.

\bibitem{Temme}
N.~Temme.
\newblock Special functions: An introduction to the classical functions of
  mathematical physics.
\newblock {\em John Wiley \& Sons}, 1996.

\bibitem{MR45253}
F.~G. Tricomi.
\newblock Asymptotische {E}igenschaften der unvollst\"{a}ndigen
  {G}ammafunktion.
\newblock {\em Math. Z.}, 53:136--148, 1950.

\bibitem{MR1675356}
H.~Widom.
\newblock On the relation between orthogonal, symplectic and unitary matrix
  ensembles.
\newblock {\em J. Statist. Phys.}, 94(3-4):347--363, 1999.

\end{thebibliography}
\end{document}